\documentclass[a4paper,UKenglish,cleveref, autoref, thm-restate]{lipics-v2021}

\usepackage{stmaryrd}

\usepackage{tikz}
\usepackage{dsfont} 
\usepackage{mathtools} 
\newcommand{\quickset}[1]{\left\lbrace #1 \right\rbrace}

\newtheorem{notation}{Notation}
\newcommand{\segment}[2]{\llbracket #1, #2 \rrbracket}

\newcounter{algorithm}
\newtheorem{algo}[algorithm]{Algorithm}{\bfseries}{\itshape}

\title{Self-Stabilization and Byzantine Tolerance for Maximal Independent Set} 

\author{Johanne Cohen}{LISN-CNRS, Universit\'e Paris-Saclay, France}{Johanne.Cohen@lri.fr}{https://orcid.org/0000-0002-9548-5260}{}

\author{ Laurence Pilard\footnote{Optional footnote, e.g. to mark corresponding author}}{LI-PaRAD,  UVSQ, Universit\'e Paris-Saclay, France,}{Laurence.Pilard@uvsq.fr}{https://orcid.org/0000-0002-1104-8216}{}

\author{Jonas Sénizergues\footnote{corresponding author}}{LISN-CNRS, Universit\'e Paris-Saclay, France}{Johanne.Cohen@lri.fr}{https://orcid.org/0000-0002-9548-5260}{}
\authorrunning{J. Cohen, L. Pilard, and  J. Sénizergues } 



\category{} 

\acknowledgements{I want to thank \dots}

\nolinenumbers 

\newcommand{\motnouveau}[1]{\emph{#1}}
\newcommand{\Nat}{\mathbb{N}}

\newcommand{\ga}{\gamma_1}
\newcommand{\gb}{\gamma_2}
\newcommand{\E}{\mathcal{E}}
\newcommand{\R}{\mathcal{R}}

\newcommand{\last}{\operatorname{last}}

\begin{document}

\maketitle

\begin{abstract} We analyze the impact of transient and Byzantine faults on the construction of a maximal independent set  in a general network. We adapt the self-stabilizing algorithm presented by Turau~\cite{turau2007linear} for computing such a vertex set. Our algorithm  is   self-stabilizing and also works under the more difficult context of arbitrary Byzantine faults. 

Byzantine nodes can prevent nodes close to them from taking part in the  independent set  for an arbitrarily long time.
We give boundaries to their impact by focusing on the set of all nodes excluding nodes at distance 1 or less of Byzantine nodes, and excluding some of the nodes at distance 2.
As far as we know, we present the first algorithm tolerating both transient and Byzantine faults under the fair distributed daemon. 

We prove that this algorithm converges in $ \mathcal O(\Delta n)$ rounds w.h.p., where $n$ and $\Delta$ are the size and the maximum degree of the network, resp. Additionally, we present a modified version of this algorithm for anonymous systems under the adversarial  distributed daemon that converges  in 
 $ \mathcal O(n^{2})$  expected number of steps.

  \end{abstract}
\section*{Introduction}

Maximal independent set has received a lot of attention in different areas.  For instance,  in wireless networks, the maximum independent sets can be used as a black box to perform communication (to collect or to broadcast information) (see~\cite{liu2017cooperative,gao2017novel}  for example).  In self-stabilizing distributed algorithms, this problem is also a fundamental tool to transform an algorithm from one model to another~\cite{grati2007, turau2006randomized}. 

 An \motnouveau{independent set} $I$ in a graph is a set of vertices such that no two of them form an edge in the graph. It is called \motnouveau{maximal} when it is maximal inclusion-wise (in which case it is also a minimal dominating set).

The maximal independent set (MIS) problem has been extensively studied in parallel and distributed settings, following the seminal works of~\cite{Alon1986,linial1987, luby1986}.  Their idea is based on the fact that a node  joins the ``MIS under construction'' $S$ according to the neighbors: node $v$ joins the set $S$ if it has no neighbor in $S$, and it leaves the set $S$ if at least one of its neighbors is in $S$.  Most algorithms in the literature, including ours, are based on this approach.

The MIS problem has been extensively studied in the {\sc Local} model, \cite{ghaffari2016improved,rozhovn2020polylogarithmic,censor2020derandomizing} for instance  (a synchronous, message-passing model of distributed computing in which messages can be arbitrarily large) and in the {\sc Congest} model~\cite{peleg2000distributed} (synchronous model where messages are $O(\log n)$ bits long). In the {\sc Local} model, Barenboim \emph{et al.}~\cite{barenboim2018locally} focus on identified system and gave a self-stabilizing algorithm producing a MIS within $O(\Delta + \log^{*} n)$ rounds.
Balliu \emph{et al}~\cite{balliu2019lower} prove that the previous algorithm~\cite{barenboim2018locally} is optimal for a wide range of parameters in the {\sc Local} model. In the {\sc Congest} model, Ghaffari \emph{et al.}~\cite{ghaffari2021improved} prove that there exists a randomized distributed algorithm that computes a maximal independent set in $O(\log \Delta \cdot \log \log n + \log^{6} \log n)$ rounds with high probability.

Self-stabilizing algorithms for maximal independent set have been designed in various models (anonymous network~\cite{shukla1995observations,turau2006randomized,turau2019making}  or not~\cite{goddard2003self,ikeda2002space}).
Up to our knowledge, Shukla \emph{et al.}~\cite{shukla1995observations} present the first algorithm designed for finding a MIS in a graph using self-stabilization paradigm for anonymous networks.  Some other self-stabilizing works deal with this problem assuming identifiers: with a synchronous daemon~\cite{goddard2003self} or distributed one~\cite{ikeda2002space}. These two works require $O(n^2)$ moves to converge. Turau~\cite{turau2007linear} improves these results to $O(n)$ moves under the distributed daemon. Recently, some works improved the results 
in the synchronous model. For non-anonymous networks,  Hedetniemi \cite{hedetniemi2021self} designed  a self-stabilization algorithm for solving the problem related to dominating sets in graphs in particular for a maximal independent set which stabilizes in $O(n)$ synchronous rounds. Moreover,  for  anonymous networks, Turau~\cite{turau2019making} 
 design some Randomized self-stabilizing algorithms for maximal independent set  w.h.p. in $O(\log  n)$ rounds.
 See the survey~\cite{guellati2010survey} for more details on MIS self-stabilizing algorithms. 
 
Some variant of the maximal independent set problem have been investigated, as for example the 1-maximal independent set problem \cite{tanaka2021self,SHI200477} or Maximal Distance-$k$ Independent Set \cite{benreguia2021selfstabilizing,johnen:hal-03138979}.   Tanaka 
\emph{et al}~\cite{tanaka2021self}  designed a silent self-stabilizing 1-MIS algorithm under the weakly-fair distributed daemon for any identified  network in  $O(nD)$ rounds (where $D$ is a diameter of the graph).

In this paper, we focus on the construction of a MIS handling both transient and Byzantine faults. 
On one side, transient faults can appear in the whole system, possibly impacting all nodes. However, these faults are not permanent, thus they stop at some point of the execution. 
Self-stabilization~\cite{dijkstra74} is the classical paradigm to handle transient faults. Starting from any arbitrary configuration, a self-stabilizing algorithm eventually resumes a correct behavior without any external intervention. On the other side, (permanent) Byzantine faults~\cite{lamport82} are located on some faulty nodes and so the faults only occur from them. However, these faults can be permanent, \emph{i.e.}, they could never stop during the whole execution.


In a distributed system, multiple processes can be active at the same time, meaning they are in a state where they could make a computation. 
The definition of self-stabilizing algorithm is centered around the notion of \motnouveau{daemon}, which captures the ways the choice of which process to schedule for the next time step by the execution environment can be made.
Two major types of daemon are the \motnouveau{sequential} and the \motnouveau{distributed} ones. A \motnouveau{sequential daemon} only allows one process to be scheduled for a given time step, while a \motnouveau{distributed daemon} allows the execution of multiple processes at the same time. 
Daemons can also be \motnouveau{fair} when they have to eventually schedule every process that is always activable, or \motnouveau{adversarial} when they are not fair.
As being distributed instead of sequential (or adversarial instead of fair) only allows for more possibilities of execution, it is harder to make an algorithm with the assumption of a distributed (resp. adversarial) daemon than with the assumption of a sequential (resp. fair) daemon.

We introduce the possibility that some nodes, that we will call \motnouveau{Byzantine} nodes, are not following the rules of the algorithm, and may change the values of their attributes at any time step. Here, if we do not work under the assumption of a fair daemon, one can easily see that we cannot guarantee the convergence of any algorithm as the daemon could choose to always activate alone the same Byzantine node, again and again. 

Under the assumption of a fair daemon, there is a natural way to express complexity, not in the number of moves performed by the processes, but in the number of \motnouveau{rounds}, where a round captures the idea that every process that wanted to be activated at the beginning of the round has either been activated or changed its mind. We give a self-stabilizing randomized algorithm that, in an 
\motnouveau{anonymous} network (which means that processes do not have unique identifiers to identify themselves) 
with Byzantine nodes under the assumption of a distributed fair daemon, finds a maximal independent set of a superset of the nodes at distance $3$ or more from Byzantine nodes. We show that  the algorithm  stabilizes in $O(\Delta n)$ rounds w.h.p., where $n$ is the size and $\Delta$ is the diameter of the underlying graph. 

In this paper, we first present the model (Section~\ref{sec:model}). Then, we give the self-stabilizing randomized algorithm with Byzantine nodes under the fair daemon  (Section~\ref{sec:Byzantine}). Finally, we  shortly describe the  
the self-stabilizing randomized algorithm in anonymous system under the adversarial daemon (Section~ \ref{sec:anonymous}).

Then, in the last part, we give a self-stabilizing randomized algorithm that finds a maximal independent set in an anonymous network, under the assumption of a distributed adversarial daemon. We show that the expected number of moves for the algorithm to stabilize is $O(n^2)$.

\section{Model}\label{sec:model}
A system consists of a set of processes where two adjacent processes can communicate with each other. The communication relation is represented by a graph $G = (V,E)$ where $V$ is the set of the processes (we will call \motnouveau{node} any element of $V$ from now on) and $E$ represents the neighbourhood relation between them, \emph{i.e.}, $uv \in E$ when $u$ and $v$ are adjacent nodes.
By convention we write $|V|=n$ and $|E|=m$. If $u$ is a node, $N(u)=\quickset{v \in V | uv \in E}$ denotes the open neighbourhood, and $N[u]=N(u) \cup \quickset{u}$ denotes the closed neighbourhood. We note $deg(u) = |N(u)|$ and $\Delta = \max \quickset{deg(u)| u\in V}$.


We assume the system to be \motnouveau{anonymous} meaning that a node has no identifier.
We use the \motnouveau{state model}, which means that each node has a set of \motnouveau{local variables} which make up the \motnouveau{local state} of the node. 
A node can read its local variables and all the local variables of its neighbours, but can only rewrite its own local variables. A \motnouveau{configuration} is the value of the local states of all nodes in the system.
When $u$ is a node and $x$ a local variable, the \motnouveau{$x$-value} of $u$ in configuration $\gamma$ is the value $x_u^{\gamma}$.

An algorithm is a set of \motnouveau{rules}, where each rule is of the form $\langle guard \rangle \rightarrow \langle command \rangle$ and is parametrized by the node where it would be applied.
The \motnouveau{guard} is a predicate over the variables of the said node and its neighbours. The \motnouveau{command} is a sequence of actions that may change the values of the node's variables (but not those of its neighbours).
A rule is \motnouveau{enabled} on a node $u$ in a configuration $\gamma$ if the guard of the rule holds on $u$ in $\gamma$. A node is \motnouveau{activable} on a configuration $\gamma$ if at least one rule is enabled on $u$.  
We call \motnouveau{move} any ordered pair $(u,r)$ where $u$ is a node and $r$ is a rule.
A move is said \motnouveau{possible} in a given configuration $\gamma$ if $r$ is enabled on $u$ in $\gamma$.

The \motnouveau{activation} of a rule on a node may only change the value of variables of that specific node, but multiple moves may be performed at the same time, as long as they act on different nodes. To capture this, we say that a set of moves $t$ is \motnouveau{valid} in a configuration $\gamma$ when it is non-empty, contains only possible moves of $\gamma$, and does not contain two moves concerning the same node.
Then, a \motnouveau{transition} is a triplet $(\gamma,t,\gamma')$ such that:
(i) $t$ is a valid set of moves of $\gamma$ and (ii) $\gamma'$ is a possible configuration after every node $u$ appearing in $t$ performed simultaneously the code of the associated rule, beginning in configuration $\gamma$.
We will write such a triplet as $\gamma \xrightarrow{t} \gamma'$. We will also write $\gamma \rightarrow \gamma'$ when there exists a transition from $\gamma$ to $\gamma'$.  $V(t)$ denotes the set of nodes that appear as first member of a couple in $t$. 

We say that a rule $r$ is \motnouveau{executed} on a node $u$ in a transition $\gamma \xrightarrow{t} \gamma'$ (or equivalently that the move $(u,r)$ is \motnouveau{executed} in $\gamma \xrightarrow{t} \gamma'$) when the node $u$ has performed the rule $r$ in this transition, that is when $(u,r) \in t$. In this case, we say that $u$ has been \motnouveau{activated} in that transition. 
Then, an \motnouveau{execution} is an alternate sequence of configurations and move sets $\gamma_0,t_1,\gamma_1 \dotsm t_i,\gamma_i, \dotsm $ where (i) the sequence either is infinite or finishes by a configuration and (ii) for all $i \in \Nat$ such that it is defined, $(\gamma_i,t_{i+1},\gamma_{i+1})$ is a transition.
We will write such an execution as $\gamma_0\xrightarrow{t_1}\gamma_1 \dotsm \xrightarrow{t_i}\gamma_i \dotsm $
When the execution is finite, the last element of the sequence is the \motnouveau{last configuration} of the execution.
An execution is \motnouveau{maximal} if it is infinite, or it is finite and no node is activable in the last configuration. It is called \motnouveau{partial} otherwise.
We say that a configuration $\gamma'$ is \motnouveau{reachable} from a configuration $\gamma$ if there exists an execution starting in configuration $\gamma$ that leads to configuration $\gamma'$. We say that a configuration is \emph{stable} if no node is activable in that configuration.

The \motnouveau{daemon} is the adversary that chooses, from a given configuration, which nodes to activate in the next transition. Two types are used: the \motnouveau{adversarial distributed daemon} 
 that allows all possible executions and  the \motnouveau{fair distributed daemon} 
that only allows executions where nodes cannot be continuously activable without being eventually activated.


Given a specification and $\mathcal{L}$ the associated set of \motnouveau{legitimate configuration}, \emph{i.e.}, the set of the configurations that verify the specification, a probabilistic algorithm is \motnouveau{self-stabilizing} when these properties are true: (\motnouveau{correctness}) every configuration of an execution starting by a configuration of $\mathcal{L}$ is in $\mathcal{L}$ and (\motnouveau{convergence}) from any configuration, whatever the strategy of the daemon, the resulting execution eventually reaches a configuration in $\mathcal{L}$ with probability~$1$.

The time complexity of an algorithm that assumes the fair distributed daemon is given as a number of \motnouveau{rounds}. The concept of round was introduced by Dolev \emph{et al.}~\cite{doismo97}, and reworded by Cournier \emph{et al.}~\cite{codevi2006} to take into account  activable nodes. We quote the two following definitions from  Cournier \emph{et al.}~\cite{codevi2006}: ``

\begin{definition}
\label{pif_disabling_action_df}
We consider that a node $u$ executes a \motnouveau{disabling action} in the transition $\ga \to \gb$ if $u$ (i) is activable in $\ga$, (ii) does not execute any rule in $\ga \to \gb$ and (iii) is not activable in $\gb$.
\end{definition}

The disabling action represents the situation where at least one neighbour of $u$ changes its local state in $\ga \to \gb$, and this change effectively made the guard of all rules of $u$ false in $\gb$. The time complexity is then computed capturing the speed of the slowest node in any execution through the round definition~\cite{doismo97}.

\begin{definition}
\label{pif_round_df}
Given an  execution $\E$, the first \emph{round} of $\E$ (let us call it $\R_1$) is the minimal prefix of $\E$ containing the execution of one action (the execution of a rule or a \emph{disabling action}) of every activable nodes from the initial configuration.
Let $\E'$ be the suffix of $\E$ such that $\E = \R_1 \E'$. The second round of $\E$ is the first round of $\E'$, and so on.''
\end{definition}
 
Observe that Definition~\ref{pif_round_df} is equivalent to Definition~\ref{my_round_df}, which is simpler in the sense that it does not refer back to the set of activable nodes from the initial configuration of the round.

\begin{definition}
\label{my_round_df}
Let $\E$ be an execution. A \emph{round} is a sequence of consecutive transitions in $\E$.
The first round begins at the beginning of $\E$; successive rounds begin immediately after the previous round has ended.
The current round ends once every node $u\in V$ satisfies at least one of the following two properties:
(i) $u$ has been activated in at least one transition during the current round or (ii) $u$ has been non-activable in at least one configuration during the current round.
\end{definition}

Our first algorithm is to be executed in the presence of \motnouveau{Byzantine} nodes; that is, there is a subset $B \subseteq V$ of adversarial nodes that are not bound by the algorithm. Byzantine nodes are always activable.  An activated Byzantine node is free to update or not its local variables. Finally, observe that in the presence of Byzantine nodes all maximal executions are infinite.
We denote by $d(u,B)$  the minimal (graph) distance between node $u$ and a Byzantine node, and we define for $i\in \mathbb{N}$:
$V_i = \quickset{u \in V | d(u,B)>i }$. 
Note that $V_0$ is exactly the set of non-Byzantines nodes, and that $V_{i+1}$ is exactly the set of nodes of $V_i$ whose neighbours are all in $V_i$.

When $i,j$ are integers, we use the 
standard mathematical notation for the integer segments:
 $\segment{i}{j}$ is the set of integers that are greater than or equal to $i$ and smaller than or equal to $j$ (\emph{i.e.} $\segment{i}{j}=[i,j]\cap \mathbb{Z} = \lbrace i, \dotsm, j \rbrace$).
$Rand(x)$ with $x \in [0,1]$ represents the random function that outputs $1$ with probability $x$, and $0$ otherwise.

\section{With Byzantines Nodes under the Fair Daemon}
\label{sec:Byzantine}

\subsection{The algorithm}

The algorithm builds a maximal independent set  represented by  a local variable~$s$. The  approach of the state of the art is the following: when two nodes are candidates to be in the independent set, then a local election  decides who will remain in the independent set. To perform a  local election,  the standard technique is to compare the identifiers of nodes. Unfortunately,  this mechanism is not robust to the presence of Byzantine nodes.

Keeping with the approach outlined above, when a node $u$ observes  that its neighbours are not in (or trying to be in) the independent set 
, the non-Byzantine node decides to join it with a certain probability.  The randomization helps to reduce the impact of Byzantine nodes. The choice of probability should reduce the impact of Byzantine nodes while maintaining the efficiency of the algorithm.

\begin{algo}
Any node $u$ has two local variables $s_u \in \quickset{\bot,\top}$ and $x_u \in \mathbb{N}$ and may make a move according to one of the following rules: \\
\textbf{(Refresh)} $x_u \not =  |N(u)| \rightarrow x_u := |N(u)| \quad (=deg(u))$ \\
\textbf{(Candidacy?)} $(x_u =  |N(u)|) \wedge (s_u=\bot)\wedge(\forall v \in N(u), s_v=\bot) \rightarrow$ \\ if $Rand(\frac{1}{1+\max(\quickset{x_v | v \in N[u]})})=1$ then $s_u := \top$\\
\textbf{(Withdrawal)} $(x_u = |N(u)|) \wedge (s_u=\top)\wedge(\exists v \in N(u), s_v =\top) \rightarrow s_u := \bot$
\end{algo}

Observe since we assume an anonymous setting, the only way to break symmetry is randomisation.  
The value of the probabilities for changing local variable must  carefully be chosen in order to reduce the impact of the  Byzantine node. 

A node joins the MIS with a probability $\frac{1}{1+\max(\quickset{x_v | v \in N[u]})}$. 
The idea to ask the neighbours about their own number of neighbours (through the use of the $x$ variable) to choose the probability of a candidacy comes from the mathematical property $\forall k \in \mathbb{N}, (1-\frac{1}{k+1})^{k} > e^{-1}$, which will allow to have a good lower bound for the probability of the event ``some node made a successful candidacy, but none of its neighbours did''.

\subsection{Specification}


Since Byzantine nodes are not bound to follow the rules, we cannot hope for a correct solution in the entire graph. What we wish to do is to find a solution that works when we are far enough from the Byzantine nodes. 
One could think about a fixed containment radius around Byzantine nodes, but as we can see later this is not as simple, and it does not work with our approach. 

Let us define on any configuration $\gamma$ the following set of nodes, that represents the already built independent set: 
$$I_{\gamma} = \quickset{u \in V_1 | (s_{u}^{\gamma}=\top) \wedge \forall v \in N(u), s_{v}^{\gamma}=\bot}$$
~\\[-1,4em]
We say that a node is \emph{locally alone} if it is candidate to be in the independent set (\emph{i.e.} its $s$-value is $\top$) while none of its neighbours are.
In configuration $\gamma$, $I_{\gamma}$ is the set of all locally alone nodes of $V_1$.   

\begin{definition}
A configuration is said \emph{legitimate} when $I_\gamma$ is a maximal independent set of $V_2 \cup I_\gamma$.
\end{definition}

\subsubsection{An example:}


\newcommand{\topo}[4]{
 \begin{tikzpicture}[scale=0.55]

\tikzstyle{vertex}=[fill=white, draw=black, shape=circle]
\tikzstyle{edge}=[-]

                 \node[style=vertex,shape =rectangle]   (b) at (-12, 0) {$b$};       
                 \node [style=vertex] (c) at (-10, 0) {$v_{1}$}; 
 	         \node [style=vertex] (d) at (-8, 0) {$v_{2}$};   		
	         \node [style=vertex] (e) at (-6, 0) {$v_{3}$};  
 
   		\draw [style=edge] (b) to (c);
 		\draw [style=edge] (c) to (d);
		 \draw [style=edge] (e) to (d);

		\draw (b.north) node[above]{#1} ;
 	         \draw (c.north) node[above]{#2} ;
 	         \draw (d.north) node[above]{#3} ;
 		\draw (e.north) node[above]{#4} ;
\end{tikzpicture}

}

Figure \ref{fig:byzantin} gives an example of an execution of the algorithm. 
Figure \ref{fig:byzantin}a depicts a network in a given configuration.  The symbol drawn above the node represents the local variable $s$. Each local variable $x$ contains the degree of its associated node. 
Byzantine node is shown with a square.  

In the initial configuration, nodes  $v_{1}$ and $v_{2}$ are  in the independent set, and then are activable for \textbf{Withdrawal}.  In the first step, the daemon activates $v_{1}$ (\textbf{Withdrawal}) and  $v_{2}$  (\textbf{Withdrawal}) leading to configuration $\gamma_1$ (Fig. \ref{fig:byzantin}b).
 In the second step, the daemon activates $v_{1}$ (\textbf{Candidacy?}). Node $v_{1}$ randomly decides whether to set $s_{v_{1}}:=\top $  leading to configuration $\gamma_2$ (Fig. \ref{fig:byzantin}c), or $s_{v_{1}}:=\bot $ leading to configuration $\gamma_1$ (Fig. \ref{fig:byzantin}b).
Assume that $v_{1}$  chooses $s_{v_{1}}:=\top $. At this moment, node $v_{1}$ is ``locally alone'' in the independent set.  In the third step, the daemon activates $b$ and  $b$ makes a Byzantine move setting $s_{b}:=\top $, leading to configuration $\gamma_3$ (Fig. \ref{fig:byzantin}d).

In the fourth step, the daemon activates $v_{1}$   (\textbf{Withdrawal}) and $b$ that sets $s_{b}:=\bot $.
  The configuration is now the same as the first configuration (Fig. \ref{fig:byzantin}b). The daemon is assumed to be fair, 
 so nodes $v_{2}$ and $v_{3}$ need to be activated before the execution can be called an infinite loop. These activations will prevent the node $v_1$ to alternate forever between in and out of the independent set, while the rest of the system remains out of it. 
  In the fifth step, the daemon activates $v_{1}$ (\textbf{Candidacy?}), $v_{2}$ (\textbf{Candidacy?}) and  $v_{3}$ (\textbf{Candidacy?}). They randomly decide to change their local variable $s$.
  Assume that $v_{1}$, $v_{2}$ and $v_{3}$ choose $s_{v_{1}}:=\top $, $s_{v_{2}}:=\bot $, and $s_{v_{3}}:=\top $, leading to configuration $\gamma_5$ (Fig. \ref{fig:byzantin}e). At this moment, node $v_{3}$ is ``locally alone'' in the independent set.  $v_3$ is far enough from the Byzantine node then it will remain in the independent set whatever $b$ does. 

\begin{figure}[h]
\centering
\begin{tabular}{p{4.1cm}p{4.1cm}p{4.1cm}}
 \topo{$\bot$}{$\top$}{$\top$}{$\bot$}  &  \topo{$\bot$}{$\bot$}{$\bot$}{$\bot$}  &  \topo{$\bot$}{$\top$}{$\bot$}{$\bot$}  \\[-1em]
~\begin{minipage}[h]{3.5cm}
(a) Initial configuration $\gamma_0$. $v_{1}$ and  $v_{2}$  execute rule \textbf{Withdrawal}. 
\end{minipage}
  & 
~\begin{minipage}[h]{3.5cm}
(b) Configuration $\gamma_1$: $v_{1}$  executes \textbf{Candidacy?}.  \\
\end{minipage}
  &
~\begin{minipage}[h]{3.5cm}
   (c) Configuration  $\gamma_2$: $b$ sets  $s_{b}:=\top $.\\
\end{minipage}
   \\[1,3em]
 \topo{$\top$}{$\top$}{$\bot$}{$\bot$}   &
\topo{$\bot$}{$\bot$}{$\bot$}{$\bot$}    &   \topo{$\bot$}{$\top$}{$\bot$}{$\top$}  
   \\[-1em]
~\begin{minipage}[h]{3.5cm}
(d) Configuration  $\gamma_3$: $v_{1}$ and $b$ change their $s$-value.\\
\end{minipage}
  & 
~\begin{minipage}[h]{3.5cm}
(e) Configuration  $\gamma_5$: Nodes $v_2$ and $v_3$ execute rule \textbf{Candidacy}
\end{minipage}
 &
 ~\begin{minipage}[h]{3.5cm}
(f) Configuration  $\gamma_6$: A $V_1$-stable configuration.
\end{minipage}
\end{tabular}
\caption{Execution}
\label{fig:byzantin}

\end{figure}

%
%
%
%



\subsubsection{About the specification:}

Our goal is to design an algorithm that builds a maximal independent set of the subgraph induced by a set of nodes where node ``too close'' to Byzantine nodes have been removed. 
The question now is to define what does ``too close'' mean. One could think about a fixed containment radius only excluding nodes at distance at most 1 from Byzantine nodes. This set of nodes has been previously defined as  $V_1$. Indeed, in Figure \ref{ex:spec}.(a),  $v_1$ and $v_2$ belongs to $V_1$ and their local view of the system is correct, then they have no reason to change their states. Moreover, Byzantine nodes are too far away to change that: whatever the value of the state of $v_0$, the view of $v_1$ remains correct. 
Thus a containment radius of 1 could seem correct. However, in Figure \ref{ex:spec}.(b), if the Byzantine node does not make any move, then $v_0$ remains in the MIS while $v_1$ remains out of it. Thus, in this example, if we only consider nodes in $V_1$, the $\top$-valued nodes of $V_1$ are not a MIS of $V_1$. If $V_1$ is not always a good choice, neither is $V_2$. See Figure \ref{ex:spec}.(c), as one can see that all nodes in $V_1$ will never change their local state. The same can be said for $V_k$ for any $k$, see example Figure \ref{ex:spec}.(d) for $V_3$. The solution is then to consider a set of nodes defined from a fixed containment radius to which we add locally alone neighboring nodes. The smallest containment radius that works with this approach is 3 (which corresponds to set $V_2$).  Note that it depends on the current configuration and not only on the underlying graph.

\subsubsection{About the choice of probability to join the MIS:} 

We could have gone with the same probability for every node, but that comes with the cost of making the algorithm very sensitive to the connectivity of the underlying graph. As we rely for convergence on the event where a node is candidate alone (\emph{i.e.} switch to $\top$ without other node doing the same in its neighborhood), the probability of progress in a given number of rounds would then be exponentially decreasing with the degrees of the graph.

We could have gone with something depending only on the degree of the node where the rule is applied. While it could have been an overall improvement over the uniform version above, the minoration of the probability of progress that can be made with a local scope is no better. We cannot exclude that a finer analysis would lead to a better overall improvement, but it would require to deal with far more complex math. On a smaller scale, we can also note that this choice would introduce a bias toward small degree nodes, while we might not want that (depending on the application).

Then, we have chosen the version where nodes takes into account their degree and what they know of the degree of their neighbors. On one hand, the first concern you could rise here would be the potential sabotage by Byzantine nodes. Here is the intuition of why this cannot be a problem here. If a node $u$ is at distance at least 2 from any Byzantine node and if $u$ is ``locally alone'' in the independent set, then whatever the Byzantine nodes do, $u$ will forever remain in the independent set. To maximize the harm done, the Byzantine nodes have to prevent indirectly such a node to join the independent set. To do so it has to maximize the probability of its neighbours to be candidate to the independent set. But Byzantine nodes cannot lie efficiently in that direction, as the probability is upper-bounded by the degree values of both the node and its non-Byzantine neighbours.
On the other hand, this choice allows us to adress the problem that we had with the previous solution. Here, we can indeed frame the probability to be candidate alone between two constant bounds with a simple local analysis. Thus, we can ensure that the convergence speed does not depend on the connectivity of the underlying graph. Again, on a smaller scale, we also greatly reduce the bias toward small degree nodes compared to the previous option.

%
%
%



\newcommand{\topoDeux}[3]{
 \begin{tikzpicture}[scale=0.55]

\tikzstyle{vertex}=[fill=white, draw=black, shape=circle]
\tikzstyle{edge}=[-]

                 \node[style=vertex,shape =rectangle]   (b) at (-12, 0) {$b$};       
                 \node [style=vertex] (c) at (-10, 0) {$v_{0}$}; 
 	         \node [style=vertex] (d) at (-8, 0) {$v_{1}$};   		
 
   		\draw [style=edge] (b) to (c);
 		\draw [style=edge] (c) to (d);

		\draw (b.north) node[above]{#1} ;
 	         \draw (c.north) node[above]{#2} ;
 	         \draw (d.north) node[above]{#3} ;
 \end{tikzpicture}
}

\newcommand{\topoTrois}[4]{
 \begin{tikzpicture}[scale=0.55]

\tikzstyle{vertex}=[fill=white, draw=black, shape=circle]
\tikzstyle{edge}=[-]

                 \node[style=vertex,shape =rectangle]   (b) at (-12, 0) {$b$};       
                 \node [style=vertex] (c) at (-10, 0) {$v_{0}$}; 
 	         \node [style=vertex] (d) at (-8, 0) {$v_{1}$};   		
	         \node [style=vertex] (e) at (-6, 0) {$v_{2}$};  

   		\draw [style=edge] (b) to (c);
 		\draw [style=edge] (c) to (d);
		 \draw [style=edge] (e) to (d);

		\draw (b.north) node[above]{#1} ;
 	         \draw (c.north) node[above]{#2} ;
 	         \draw (d.north) node[above]{#3} ;
	         \draw (e.north) node[above]{#4} ;

 \end{tikzpicture}
}

\newcommand{\topoCinq}[6]{
 \begin{tikzpicture}[scale=0.55]

\tikzstyle{vertex}=[fill=white, draw=black, shape=circle]
\tikzstyle{edge}=[-]

                 \node[style=vertex,shape =rectangle]   (b) at (-12, 0) {$b$};       
                 \node [style=vertex] (c) at (-10, 0) {$v_{0}$}; 
 	         \node [style=vertex] (d) at (-8, 0) {$v_{1}$};   		
	         \node [style=vertex] (e) at (-6, 0) {$v_{2}$};  
 	         \node [style=vertex] (f) at (-4, 0) {$v_{3}$};  
 	         \node [style=vertex] (g) at (-2, 0) {$v_{4}$};  

   		\draw [style=edge] (b) to (c);
 		\draw [style=edge] (c) to (d);
		 \draw [style=edge] (e) to (d);
		 \draw [style=edge] (e) to (f);
		 \draw [style=edge] (f) to (g);

		\draw (b.north) node[above]{#1} ;
 	         \draw (c.north) node[above]{#2} ;
 	         \draw (d.north) node[above]{#3} ;
	         \draw (e.north) node[above]{#4} ;
 		 \draw (f.north) node[above]{#5} ;
 		 \draw (g.north) node[above]{#6} ;

 \end{tikzpicture}
}
\newcommand{\topoSix}[7]{
 \begin{tikzpicture}[scale=0.55]

\tikzstyle{vertex}=[fill=white, draw=black, shape=circle]
\tikzstyle{edge}=[-]

                 \node[style=vertex,shape =rectangle]   (b) at (-12, 0) {$b$};       
                 \node [style=vertex] (c) at (-10, 0) {$v_{0}$}; 
 	         \node [style=vertex] (d) at (-8, 0) {$v_{1}$};   		
	         \node [style=vertex] (e) at (-6, 0) {$v_{2}$};  
 	         \node [style=vertex] (f) at (-4, 0) {$v_{3}$};  
 	         \node [style=vertex] (g) at (-2, 0) {$v_{4}$};  
  	         \node [style=vertex] (h) at (0, 0) {$v_{5}$};  

   		\draw [style=edge] (b) to (c);
 		\draw [style=edge] (c) to (d);
		 \draw [style=edge] (e) to (d);
		 \draw [style=edge] (e) to (f);
		 \draw [style=edge] (f) to (g);
 		 \draw [style=edge] (g) to (h);

		\draw (b.north) node[above]{#1} ;
 	         \draw (c.north) node[above]{#2} ;
 	         \draw (d.north) node[above]{#3} ;
	         \draw (e.north) node[above]{#4} ;
 		 \draw (f.north) node[above]{#5} ;
 		 \draw (g.north) node[above]{#6} ;
 		 \draw (h.north) node[above]{#7} ;

 \end{tikzpicture}
}

\begin{figure}[h]
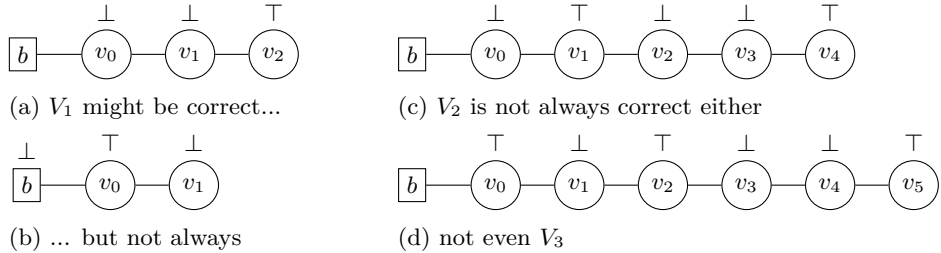

\label{ex:spec}
\begin{tabular}{p{4.7cm}p{6.1cm}}
\topoTrois{}{$\bot$}{$\bot$}{$\top$} & 
\topoCinq{}{$\bot$}{$\top$}{$\bot$}{$\bot$}{$\top$}\\
(a) $V_1$ might be correct...  & (c) $V_2$ is not always correct either
\end{tabular}
~\\
\begin{tabular}{p{4.7cm}p{6.1cm}}

\topoDeux{$\bot$}{$\top$}{$\bot$} & 
\topoSix{}{$\top$}{$\bot$}{$\top$}{$\bot$}{$\bot$}{$\top$}\\
(b) ... but not always & (d) not even $V_3$
\end{tabular}
\caption{What is the good containment radius?}

\end{figure}

\subsection{The proof}

%
%
%
%

Every omitted proof can be seen in the appendix.

We say that in a configuration $\gamma$, a node $u$ is \motnouveau{degree-stabilized} if rule \textbf{Refresh} is not enabled on it. The configuration  $\gamma$ is then said to be \motnouveau{degree-stabilized} if every non-Byzantine node is degree-stabilized.
Observe the two following facts about degree-stabilization (proofs in appendix):
\begin{lemma}
Any reachable configuration from a degree-stabilized configuration is degree-stabilized.
\end{lemma}


\begin{lemma} \label{bz-lem-stab-1round}
From any configuration $\gamma$, the configuration $\gamma'$ after one round is degree-stabilized.
\end{lemma}

%
%


%
%


All  locally alone nodes in $V_{1}$ (\emph{i.e.} nodes of $I_\gamma$) remains  locally alone during the whole execution. 

\begin{lemma} \label{bz-lem-I-inc}

If $\gamma \rightarrow \gamma'$, $I_{\gamma} \subseteq I_{\gamma'}$.
\end{lemma}

%
%

We now focus on the progression properties of enabled rules after one round.  We start with the
\textbf{Withdrawal} rule.  When the  \textbf{Withdrawal} rule is enabled on a node $u \in V_{1}$,  the conflict is solved by either making $u$ locally alone  or setting $s_{u} = \bot $.

\begin{lemma} \label{bz-lem-withdrawal}
If $\gamma$ is a degree-stabilized configuration and if \textbf{Withdrawal} is enabled on $u\in V_1$ 
then after one round either \textbf{Withdrawal} has been executed on $u$ or in the resulting configuration $\gamma'$ we have $u \in I_{\gamma'}$.
\end{lemma}

\begin{proof}
Since $\gamma$ is degree-stabilized, no \textbf{Refresh} move can be executed in any future transition. Then, since $s_u^{\gamma}=\top$, only \textbf{Withdrawal} can be executed on $u$ or any of its non-Byzantine neighbours until $u$ has been activated. Since $u \in V_1$, it is in fact true for every neighbour of $u$.


\noindent Since $u$ is activable in $\gamma$, we have two cases:
(i) If $u$ has performed a \textbf{Withdrawal} move in the next round there is nothing left to prove.
(ii) If it is not the case, $u$ must have been unactivated by fairness hypothesis, that means that each neighbour $v\in N(u)$ that had $s$-value $\top$ in $\gamma$ have been activated. By the above, they must have performed a \textbf{Withdrawal} move that changed their $s$-value to $\bot$. Also $u$ is supposed not to have performed any rule, so it keeps $s$-value $\top$ in the whole round: its neighbours -that are non-Byzantine since $u\in V_1$- cannot perform any \textbf{Candidacy} move. As such, in the configuration $\gamma'$ at the end of the round, every neighbour of $u$ has $s$-value $\bot$, and $u$ has $s$-value $\top$. Since $u\in V_1$ that means that $u \in I_{\gamma'}$. 
\end{proof}

When    \textbf{Candidacy?} rule is enabled on a node $u \in V_{1}$, then \textbf{Candidacy?} is executed on $v \in N[u]$ within one round.

\begin{lemma} \label{bz-lem-candidacyworks}
If in $\gamma$ degree-stabilized we have the \textbf{Candidacy?} rule enabled on $u \in V_1$ 
 (\emph{i.e.}, $s_{u}^{\gamma} = \bot$ and $\forall v \in N(u), s_{v}^{\gamma} = \bot$) then after one round \textbf{Candidacy?} have either been executed on $u$, or on at least one neighbour of $u$.
\end{lemma}

\begin{proof}
Since $\gamma$ is degree-stabilized, no \textbf{Refresh} move can be executed in any future transition. Then, until $u$ or one of its neighbours have been activated, only \textbf{Candidacy?} can be executed on them since it's the only rule that can be activated on a node with $s$-value $\bot$.

\noindent Since $u$ is activable in $\gamma$, by fairness, we have two cases:

(i) If $u$ is activated before the end of the round, the only rule that it could have performed for its first activation is the \textbf{Candidacy?} rule since its $s$-value in $\gamma$ is $\bot$ and the configuration is supposed degree-stabilized.

(ii) If not, it has been unactivated, which means that at least one neighbour $v\in N(u)$ has been activated. As $u\in V_1$, $v$ cannot be Byzantine. The only rule that it could have performed for its first activation is the \textbf{Candidacy?} rule since its $s$-value in $\gamma$ is $\bot$ and the configuration is supposed degree-stabilized. 
\end{proof}

\noindent If  node $u\in V_{1}$ executes   \textbf{Candidacy?} rule, then $u$ becomes a locally alone node with  a certain probability in the next configuration.  So it implies that set $I$~grows.

\begin{lemma}\label{bz-lem-main}
If $\gamma$ is a degree-stabilized configuration, and in the next transition rule \textbf{Candidacy?} is executed on a node $u\in V_1$ 
, there is probability at least $\frac{1}{e(\Delta+1)}$ that in the next configuration $\gamma'$, $s_u^{\gamma'}=\top$ and $\forall v\in N(u), s_v^{\gamma'}=\bot$. 
\end{lemma}

\begin{proof}
For any node $y$, we write $\varphi(y) =  \frac{1}{1+max(\quickset{deg(v) | v \in N[y]})}$.

Since \textbf{Candidacy?} can be executed on $u$ in $\gamma$, we know that $\forall v \in N[u], s_{v}^{\gamma} = \bot$. Since $\gamma$ is degree-stabilized and no neighbour of $u$ can be Byzantine by definition of $V_1$, we have $\forall v \in N[u], x_{v}^{\gamma} = deg(v)$. The probability of $s_u^{\gamma'}=\top$ knowing the rule has been executed is then $\varphi(u)$.

Then, for a given $v \in N(u)$, node $v$ is not Byzantine since $u \in V_1$, and the probability that $s_v^{\gamma'} = \top$ is either $0$ (if \textbf{Candidacy?} has not been executed on $v$ in the transition) or $\frac{1}{1+\max(\quickset{x_w^{\gamma} | w \in N[v]})}$. Since $deg(u)= x_u \leq \max(\quickset{x_w^{\gamma} | w \in N[v]})$, that probability is then at most $\frac{1}{1+deg(u)}$.

Thus (since those events are independents), the probability for $u$ to be candidate in $\gamma'$ without candidate neighbour is at least 
\\[-1em]
$$p = \varphi(u)\prod_{v \in N(u)} \left(1-\frac{1}{1+deg(u)} \right) \geq \frac{1}{\Delta+1} \left(1-\frac{1}{1+deg(u)} \right)^{deg(u)}$$
\\[-,5em]
\noindent Then, as $\forall k \in \mathbb{N}, (1-\frac{1}{k+1})^{k} > e^{-1}$ (proof in appendix),\\ 
$$\hspace*{8cm} p > \frac{1}{\Delta+1} \times \frac{1}{e}$$ 
 and the lemma holds. 

\end{proof}


\begin{lemma} \label{bz-lem-3way}
If $\gamma$ is a degree-stabilized configuration such that $I_{\gamma}$ is not a maximal independent set of $V_2 \cup I_{\gamma}$, then after at most one round one of the following events happens:\\
\begin{enumerate}
\item Rule \textbf{Candidacy?} is executed on a node of $V_1$
\item A configuration $\gamma'$ such that $I_{\gamma} \subsetneq I_{\gamma'}$ is reached.
\item A configuration $\gamma'$ such that rule \textbf{Candidacy?} is enabled on a node of $V_2$ in $\gamma'$ is reached.
\end{enumerate}
\end{lemma}

\begin{proof}
Suppose $\gamma$ is a degree-stabilized configuration such that $I_{\gamma}\cup V_2$ is not a maximal independent set of $V_2$. As $\gamma$ is supposed degree-stabilized, we will only consider the possibility of moves that are not \textbf{Refresh} moves.
(a) If \textbf{Candidacy?} is enabled on a node of $V_2$ in $\gamma$, Condition~$3$ holds.
(b) If it is not the case, then there exists $u \in V_2$ that has at least one neighbour $v \in V_1$ such that $s_{u}^{\gamma} = s_{v}^{\gamma} = \top$  (otherwise $I_{\gamma}$ would be a maximal independent set of $V_2 \cup I_{\gamma}$).

In the first case there is nothing left to prove. In the second case, \textbf{Withdrawal} is enabled on $u \in V_2$ in $\gamma$ and from Lemma~\ref{bz-lem-withdrawal}, we have two possible cases: (i) if $\gamma'$ is the configuration after one round, $u \in I_{\gamma'}$ and Condition~$2$ holds ; (ii) $u$ perform a \textbf{Withdrawal} move before the end of the round.

In the first case there is nothing left to prove. Suppose now that we are in the second case and that $u$ is activated only once before the end of the round, without loss of generality since Condition~$1$ would hold otherwise as it second activation would be a \textbf{Candidacy?} move and $u \in V_1$. 
Then: \\
\begin{itemize} 
\item If within a round a configuration $\gamma'$ is reached where a node $w \in N[u]$ is such that $s_w^{\gamma'} = \top$ and $w$ is not activable we have: $w \not\in I_{\gamma}$ (as $u$ is $\top$-valued in $\gamma$) which gives $I_{\gamma} \subsetneq I_{\gamma} \cup \quickset{w} \subseteq I_{\gamma'}$ by Lemma~\ref{bz-lem-I-inc}, thus Condition~$2$ holds.

\item If we suppose then that no such event happens until the end of the round in configuration $\gamma'$, we are in either of those cases:
(i)  Every neighbour of $u$ have value $\bot$ in $\gamma'$ and and $s_u^{\gamma'} = \bot$ (as we would be in the previous case if it was $\top$), thus \textbf{Candidacy?} is enabled on $u$ in $\gamma'$ and Condition~$3$ holds.
(ii)  A \textbf{Candidacy?} has been performed within the round on a neighbour of $u$ and Condition~$1$ holds.
\end{itemize}
~\\
Thus, in every possible case, one of the three conditions holds. 
\end{proof}
~\\
\noindent  Every 2 rounds, the set of locally alone\,nodes strictly\,grows with\,some probability.

\begin{lemma} \label{bz-lem-advance}
If $\gamma$ is degree-stabilized, with $I_{\gamma}$ not being a maximal independent set of $V_2 \cup I_{\gamma}$, then after two rounds the probability for the new configuration $\gamma'$ to be such that $I_{\gamma} \subsetneq I_{\gamma'}$ is at least $\frac{1}{(\Delta+1)e}$.
\end{lemma}

\begin{proof}

From Lemma~\ref{bz-lem-3way}, we have three possibilities after one rounds.\\
\begin{itemize}
\item If we are in Case~$1$, let us denote by $\gamma''$ the resulting configuration after the transition where the said \textbf{Candidacy?} move have been executed on $u\in V_1$. Then using Lemma~\ref{bz-lem-main}, we have $u \in I_{\gamma''}$ with probability at least $\frac{1}{(\Delta+1)e}$. Since we have $u\not\in I_{\gamma}$  (if $u$ was in $I_{\gamma}$, \textbf{Candidacy?} could not have been executed on $u$ after configuration $\gamma$) and by Lemma~\ref{bz-lem-I-inc}, we have then $I_{\gamma} \subsetneq I_{\gamma'}$ with probability at least $\frac{1}{(\Delta+1)e}$.
\item If we are in Case~$2$, there is nothing left to prove.
\item If we are in Case~$3$, let us denote by $\gamma''$ the first configuration where \textbf{Candidacy?} is enabled on some $u\in V_2$. Then using Lemma~\ref{bz-lem-candidacyworks} after at most one more round \textbf{Candidacy?} will be executed on $v\in N[u]$. Let us denote by $\gamma''$ the resulting configuration after the transition where the said \textbf{Candidacy?} move have been executed on $v$. Since $u\in V_2$ we have $v\in V_1$ and using Lemma~\ref{bz-lem-main} $v \in I_{\gamma'''}$ with probability at least $\frac{1}{(\Delta+1)e}$ and $v \not\in I_{\gamma}$ by the same argument as above. Thus, since $I_{\gamma'''} \subseteq I_{\gamma'}$, the probability that $I_{\gamma} \subsetneq I_{\gamma'}$ is at least $\frac{1}{(\Delta+1)e}$.
\end{itemize}
~\\
In every case, the property is true, thus the lemma holds. 
\end{proof}

Since the number of locally alone nodes cannot be greater than $n$, the expected number of rounds before stabilization can be computed.
%
%
%
We will use the notation $\alpha = \frac{1}{(\Delta+1)e}$ to simplify the formulas in the remaining of the paper.

\begin{lemma} \label{bz-lem-speed}
For any $p\in [0,1[$. From any degree-stabilized configuration $\gamma$, Algorithm  is self-stabilizing for  a configuration $\gamma'$ where $I_{\gamma'}$ is a maximal independent set of $V_2 \cup I_{\gamma'}$, with time complexity  $\max\left( -\alpha^2\ln p,\frac{\sqrt{2}}{\sqrt{2}-1}\frac{n}{\alpha} \right)$ rounds
with probability at least $1-p$.
\end{lemma}

\begin{theorem}
For any $p\in [0,1[$. From any  configuration $\gamma$, Algorithm is self-stabilizing for a configuration $\gamma'$ where $I_{\gamma'}$ is a maximal independent set of $V_{2} \cup  I_{\gamma'}$. 
%
\end{theorem}
%
\setcounter{section}{2}

\section{In an Anonymous System under the Adversary Daemon}
\label{sec:anonymous}

\subsection{The algorithm}

The algorithm builds a maximal independent set represented by a local variable~$s$.

Note that we could have used the previous algorithm as fairness was only needed to contain byzantine influence, but the complexity would have been something proportional to $\Delta n^2$, and as we will prove we can do better than that.

We keep the idea of having nodes making candidacy, and then withdraw if the situation to be candidate is not right. As we still do not have identifiers, we also need probabilistic tie-break. But contrary to the byzantine case, we move the probabilities to the \textbf{Withdrawal} rule: a non-candidate node with no candidate neighbour will always become candidate when activated, but a candidate node with a candidate neighbour will only withdraw with probability $\frac{1}{2}$ when activated.
\begin{algo}
Any node $u$ has a single local variable $s_u \in \quickset{\bot,\top}$ and may make a move according to one of the following rules: \\
\textbf{(Candidacy)} $(s_u=\bot)\wedge(\forall v \in N(u), s_v = \bot) \rightarrow s_u := \top$\\
\textbf{(Withdrawal?)} $(s_u=\top)\wedge(\exists v \in N(u), s_v =\top) \rightarrow $ if $Rand(\frac{1}{2})=1$ then $s_u := \bot$\\
\end{algo}

The idea behind this is that we can give a non-zero lower bound on the probability that a connected component of candidate nodes eventually collapses into at least one definitive member of the independent set. As every transition with \textbf{Candidacy} moves makes such sets appear, and \textbf{Withdrawal?} moves make those collapse into member of the independent set, it should converge toward a maximal independent set.

We will not in fact use exactly connected component, as it was more handy in the proof to consider set on nodes that became candidate as the same time, but that's where the intuition comes from.

Given a configuration $\gamma$, we define $\beta(\gamma) = \quickset{u \in V | s_u {=} \top \wedge \forall v \in N(u), s_v {=} \bot}$. Note that $\beta(\gamma)$ is always an independent set since two distinct members cannot be neighbours (as they have both $s$-value $\top$).

\subsection{An example}


\newcommand{\topologie}[4]{
\tikzstyle{vertex}=[fill=white, draw=black, shape=circle]
\tikzstyle{edge}=[-]

                 \node[style=vertex]   (a) at (-12, 8) {$a$};       
                 \node [style=vertex] (b) at (-10, 8) {$b$}; 
 	         \node [style=vertex] (c) at (-11, 6) {$c$};   		
	         \node [style=vertex] (d) at (-9, 6) {$d$};  
 
 		\draw [style=edge] (a) to (b);
 		\draw [style=edge] (a) to (c);
 		\draw [style=edge] (b) to (c);
 		\draw [style=edge] (c) to (d);
		\draw (a.north) node[above]{#1} ;
 	         \draw (b.north) node[above]{#2} ;
 	         \draw (c.north) node[above]{#3} ;
 		\draw (d.north) node[above]{#4} ;
}
The aim  of the algorithm is to build a set independent represented by $\beta(\gamma)$.  The approach of the algorithm is the following: when a node is in the set independent it remains so throughout the execution.

Below is an execution of the algorithm under the adversarial distributed daemon. Figure \ref{fig:ex}a shows the initial configuration $\gamma_{0}$ of the execution. Node identifiers are indicated inside the circles.  The symbols $\bot$ and $\top$ show the content of the local variable $s$. Consider the initial configuration (Figure \ref{fig:ex}a) in which all local variables are equal to $\bot$.
Since no node is in set independent in   configuration $\gamma_{0}$, 
the \textbf{Candidacy} rule is firable on all  nodes.  During the transition $\gamma_{0} \to  \gamma_{1}$,  the demon activates all the nodes and it means that for any node $ (u,\textbf{Candidacy})$ is a move of transition $t_{1}$. The configuration $\gamma_{1}$ is drawn in Figure \ref{fig:ex}b.  The possible moves are $(u,\textbf{Withdrawal?})$ for any node in $G$. The daemon can choose all the nodes that can 
be executed. 

In the transition $t_{2}$, moves $ (a,\textbf{Withdrawal?})$ and $ (b,\textbf{Withdrawal?})$ are executed : using our notation,
$t_{2}=\quickset{(a,\textbf{Withdrawal?}),(b,\textbf{Withdrawal?})}$.
Depending on the random drawing, none of these moves change the $s$-values.   The configurations $\gamma_1$ and $\gamma_2$ are the same. 

In the transition $t_{3}$, moves $ (a,\textbf{Withdrawal?})$, $ (b,\textbf{Withdrawal?})$ and\\ $ (c,\textbf{Withdrawal?})$ are executed. Depending the random choices,  only node $c$ changes its $s$-value (see Figure \ref{fig:ex}c). 
Observe that from this configuration $\gamma_3$,   node $d$ is an element  of independent set, and will remain so until the end of the execution: $\beta(\gamma_3)=\quickset{d}$. Then, it can not be executed no more moves.

In the transition $t_{4}$, we assume that move  $ (a,\textbf{Withdrawal?})$, \\$(b,\textbf{Withdrawal?})$ are executed. According to the random choice, nodes $a$ and $b$ change their $s$-value (see Figure \ref{fig:ex:1}d).  Observe that
$c$ has $s$-value $\bot$ and cannot be in a candidate set. $\quickset{a,b,d}$, $\quickset{a,b}$ and $\quickset{d}$ are candidate set for that configuration, but not $\quickset{a}$, $\quickset{b}$, $\quickset{a,d}$ or $\quickset{b,d}$.
 
Again, in the transition $t_{5}$  moves $ (a,\textbf{Candidacy})$, and $ (b,\textbf{Candidacy})$ are executed  (see Figure \ref{fig:ex:1}e).   The timed move $(5,(a,\textbf{Candidacy}))$ is labelled by $(5,\quickset{a,b,d})$ and $(5,\quickset{a,b})$. 
Since $A_5= \quickset{a,b,d}$, the color of  this timed move   is $(5,\quickset{a,b,d})$.

Afterward, moves $ (a,\textbf{Withdrawal?})$, $ (b,\textbf{Withdrawal?})$ are only the possible moves. All theses possibles has  $(5,\quickset{a,b,d})$ as their color.  

In the next two transitions $t{6}$ and $t{7}$, the timed moves $(a,\textbf{Withdrawal?})$ and $(b,\textbf{Withdrawal?})$ are executed without modifying their $s$-values. Their  moves is always color of $(5,\quickset{a,b,d})$.  Then, in the transition $t_{8}$ only  node $a$  executes the rule \textbf{Withdrawal?}. Finally,  it changes its $s$-value. The $\gamma_{8}$ configuration is stable since no node is elligible to execute a rule.

\begin{figure} 
\centering
\begin{tabular}{p{4cm}p{4cm}p{4cm}}
\begin{tikzpicture}
\topologie{$\bot$}{$\bot$}{$\bot$}{$\bot$} ;
\end{tikzpicture} & \begin{tikzpicture}
\topologie{$\top$}{$\top$}{$\top$}{$\top$} ;
\end{tikzpicture} &\begin{tikzpicture}
\topologie{$\top$}{$\top$}{$\bot$}{$\top$} ;
\end{tikzpicture} \\
~\begin{minipage}[h]{3.5cm}
(a) Initial configuration $\gamma_0$. All nodes execute rule \textbf{Candidacy}
\end{minipage}
& 
~\begin{minipage}[h]{3.5cm}
(b) Configuration $\gamma_1=\gamma_2$ : All nodes can execute rule  \textbf{Withdrawal?}  
\end{minipage}
& 
~\begin{minipage}[h]{3.5cm}
(c) Configuration  $\gamma_3$: Some nodes were activated, but only node $c$ had its $s$-value modified 
\end{minipage}
~\\ 
\begin{tikzpicture}
\topologie{$\bot$}{$\bot$}{$\bot$}{$\top$} ;
\end{tikzpicture} & \begin{tikzpicture}
\topologie{$\top$}{$\top$}{$\bot$}{$\top$} ;
\end{tikzpicture} &   \begin{tikzpicture}
\topologie{$\bot$}{$\top$}{$\bot$}{$\top$} ;
\end{tikzpicture}  \\
~\begin{minipage}[h]{3.5cm}
(d) Configuration  $\gamma_4$ : Nodes $a$ and $b$ had their $s$-value changed using rule \textbf{Withdrawal?} 
\end{minipage}
&
~\begin{minipage}[h]{3.5cm}
(e) Configuration  $\gamma_5=\gamma_6=\gamma_7$ : Nodes $a$ and $b$ had their $s$-value changed using rule \textbf{Candidacy} 
\end{minipage}
&
~\begin{minipage}[h]{3.5cm}
(f) Configuration  $\gamma_8$: Final configuration, as no node may perform any rule
\end{minipage}
\end{tabular}
\label{fig:ex}\label{fig:ex:1}

\end{figure}

%
%
%
%


\subsection{The proof}

\begin{lemma} \label{betamax}
In any stable configuration $\gamma$, $\beta(\gamma)$ is a maximal independent set of $G$.
\end{lemma}

\begin{proof}
Observe that  $\beta(\gamma)$ is an independent set of $G$ because two distinct members cannot be neighbours (as they have both $s$-value $\top$). 

Now suppose by contradiction that $\beta(\gamma)$ is not maximal, \emph{i.e.} that there exists  a node $u\not\in\beta(\gamma)$  that has no neighbour in $\beta(\gamma)$. We now search for a node $v$ such that $v \not\in\beta(\gamma)$ and $s_v^{\gamma}=\top$. To do this, we consider two cases according to the $s$-value of $u$:
\begin{itemize}
\item If $s_u^{\gamma}=\bot$ then,  \textbf{Candidacy} cannot be enabled on $u$ in $\gamma$ by stability of $\gamma$.  Node $u$ must have at least one neighbour $v$ such that $s_v^{\gamma}=\top$. Thus,  $v \not\in\beta(\gamma)$ and $s_v^{\gamma}=\top$. 
\item If $s_u^{\gamma}=\top$, then in this case, since $u\not\in\beta(\gamma)$, $v$ is  node  $u$.
\end{itemize}
Then, in both cases, there exists a node $v$   such that $v \not\in\beta(\gamma)$ and $s_v^{\gamma}=\top$.
Since $v \not\in\beta(\gamma)$ and $s_v^{\gamma}=\top$, by definition of $\beta(\gamma)$, $v$ must have a neighbour $w$ such that $s_w^{\gamma}=\top$, and then \textbf{Withdrawal?} would be enabled on $v$, contradiction with the stability of $\gamma$. \hfill \
\end{proof}

\begin{lemma} \label{betamax2}
In any configuration $\gamma$, if $\beta(\gamma)$ is a maximal independent set of $G$, then $\gamma$ is stable.
\end{lemma}

\begin{proof}
Suppose $\beta(\gamma)$ is a maximal independent set of $G$.
Suppose by contradiction that $\gamma$ is not stable, \emph{i.e.} that one of the rules is enabled on some node $u$. Then:
\begin{itemize}
\item If \textbf{Candidacy} is enabled on a node $u$ in $\gamma$, then   $s_u = \bot$ and $u\not\in\beta(\gamma)$. Then $u$ has only neighbours with $s$-value $\bot$ thus $u$ has no neighbour in $\beta(\gamma)$. Thus $\beta(\gamma) \cup \quickset{u}$ is also an independent set, bigger than $\beta(\gamma)$. This is a contradiction with the maximality of $\beta(\gamma)$.
\item If \textbf{Withdrawal?} is enabled on a node $u$ in $\gamma$, then $s_u=\top$ with a neighbour with $s$-value $\top$ thus $u\not\in\beta(\gamma)$ by definition. Then  and   no neighbour of $u$ can    be  in $\beta(\gamma)$ by definition  because $s_u= \top$.  This implies that $\beta(\gamma) \cup \quickset{u}$ is a bigger independent set than $\beta(\gamma)$: contradiction with the maximality of $\beta(\gamma)$.
\end{itemize}
\ \end{proof}

We synthetize results from Lemma~\ref{betamax} and Lemma~\ref{betamax2} as follows:

\begin{corollary} \label{adv-cor-stablecaracterization}
$\gamma$ is a stable configuration if and only if $\beta(\gamma)$ is a maximal independent set.
\end{corollary}

Now that we know that if $\beta(\gamma)$ is a maximal independent set the configuration is stable, we need to prove that it does grow to be maximal. To do this, we begin by proving that it cannot lose members.

Lemma~\ref{adv-lem-betagrows} means that a node in the independent set represented by $\beta(\gamma)$ remains in this set in the futur.
\begin{lemma} \label{adv-lem-betagrows}
If $\gamma \rightarrow \gamma'$, then $\beta(\gamma) \subseteq \beta(\gamma')$.
\end{lemma}

\begin{proof}
If $u$ is in $\beta(\gamma)$, we have by definition:
\begin{enumerate}
\item $s_u^{\gamma}=\top$
\item Every neighbour of $u$ have $s$-value $\bot$ in $\gamma$
\end{enumerate}
Point~1 implies that \textbf{Candidacy} is not enabled on $u$ in $\gamma$ nor on any neighbour of $u$. Moreover, due to  Point~2,  \textbf{Withdrawal?} is not enabled on $u$ in $\gamma$ nor on any neighbour of $u$.

Since no rule is enabled on $u$ or on its neighbours, the $s$-values of $u$ and its neighbours is the same in $\gamma$ and $\gamma'$ and thus $u \in \beta(\gamma')$.
\ \end{proof}

%

Now that we know that $\beta$ can only grow, it remains to prove that it does within a finite period of time. We are introducing the following concepts to this end:

$A \subseteq V$ is said to be a \motnouveau{candidate set} of a configuration $\gamma$ if $$\forall u \in A, (s_u^{\gamma} = \top) \wedge \left( \forall v\in N(u),  (s_v^{\gamma} = \bot) \vee (v \in A) \right).$$ Note that there can be multiple candidate set of a given configuration, for example $\emptyset$ and the set of all nodes with $s$-value $\top$ are always candidate sets.

\begin{figure}[h]
\begin{center}
\begin{tikzpicture}
\tikzstyle{vertex}=[fill=white, draw=black, shape=circle]
\tikzstyle{edge}=[-]
                  \node[style=vertex]   (0) at (-12, 8) {$a$};       \draw (0.north) node[above]{$\top$} ;
 	    	\node [style=vertex] (1) at (-10, 8) {$b$};  \draw (1.north) node[above]{$\top$} ;

 	         \node [style=vertex] (2) at (-11, 6) {$c$}; \draw (2.north) node[above]{$\bot$} ;
 		\node [style=vertex] (3) at (-9, 6) {$d$}; \draw (3.north) node[above]{$\top$} ;

 		\draw [style=edge] (0) to (1);
 		\draw [style=edge] (0) to (2);
 		\draw [style=edge] (1) to (2);
 		\draw [style=edge] (2) to (3);
\end{tikzpicture}
\caption{$c$ has $s$-value $\bot$ and cannot be in a candidate set. $\quickset{a,b,d}$, $\quickset{a,b}$ and $\quickset{d}$ are candidate set for that configuration, but not $\quickset{a}$, $\quickset{b}$, $\quickset{a,d}$ or $\quickset{b,d}$.}
\end{center}
\end{figure}
 
In the remaining of this part, we will almost always talk about a given execution  $\gamma_0 \xrightarrow{t_1} \gamma_1 \dotsm \gamma_{i-1} \xrightarrow{t_i} \gamma_{i}\dotsm$ that is supposed to be ``complete'' in the sense that it either stops in a stable configuration, or is infinite. The notations $\gamma_i$ and $t_i$ will be used in reference to this when not stated otherwise. In some cases though, we will not be able to do it that way. When we have to reason about probabilities of events to happen, we will need to reason on \motnouveau{partial} execution $\gamma_0 \xrightarrow{t_1} \gamma_1 \dotsm \gamma_{i-1} \xrightarrow{t_i} \gamma_{i}$, whose future has yet to be defined. It will be made explicit when we do so.

In order to be able to count the total number of moves, we aim to separate the performed moves into easily countable classes. For this, we define the \motnouveau{timed moves}, that are couples formed by the index of a transition and a move the transition of that index: $(i,\delta)$ is a timed move when $\delta \in t_i$.
Then, we want to identify the candidates set to which are connected each timed move in the following manner:

Let $A$ be a candidate set of $\gamma_i$, $u\in A$, and $j\in \Nat$ such that $i < j$.
\begin{itemize}
\item When $\delta= (u, \textbf{Withdrawal?})$ is a move in $t_j$, we say that the timed move $(j,\delta)$ is \motnouveau{labelled} $(i,A)$ if $\forall k\in \segment{i}{j-1} , s_u^{\gamma_k} = \top$. We can understand this definition by saying that $\delta$, as a move of the $j$-th transition, comes from the existence of the candidate set $A$ in the configuration $\gamma_i$. 
\item When $i>0$, if $\delta = (u, \textbf{Candidacy})$ is a move in $t_{i}$, we say that the timed move $(i,\delta)$ is \motnouveau{labelled} $(i,A)$. We can understand this by saying that $\delta$, as a move of the $i$-th transition, is a move responsible for the apparition of $A$ as a candidate set of $\gamma_i$.
\end{itemize}

There, you may notice that one given move can have multiple label as a move of a given transition. For example, for $A$ and $A'$ two candidate set of $\gamma_i$ such that $A\subseteq A'$, if $(j,\delta)$ has label $(i,A')$, it has also label $(i,A)$.
As we are interested in a partition of the effective move of an execution not to count the same move multiple times, we will define a specialized version of the labels, using only the specific candidate set, defined as follows:

For each configuration $\gamma_i$ with $i>0$, we take \\$A_i = \quickset{ u \in V | (s_u^{\gamma_{i-1}} = \bot) \wedge (s_u^{\gamma_{i}} = \top)}$. Due the guard of the \textbf{Candidacy} rule, one can verify that it is exactly the set of nodes $u$ such that $(u,\textbf{Candidacy}) \in t_i$. We also define $A_0 = \quickset{ u \in V | s_u^{\gamma_{0}} = \top}$.

\begin{lemma} \label{adv-lem-Aicandidate}
For any $i \in \Nat$, $A_i$ is a candidate set of the configuration $\gamma_i$.
\end{lemma}

\begin{proof}
By definition of $A_0$, it contains every node with value $\top$ in $\gamma_0$, and must then be a candidate set.

If $i\not=0$, let's take $u \in A_i$, and $v \in N(u)$. By definition of $A_i$, $(u,\textbf{Candidacy}) \in t_i$. Since $t_i$ is a valid set of move, we have $s_v^{\gamma_{i-1}}=\bot$. Then we have either $(v,\textbf{Candidacy}) \in t_i$ and then $v \in A_i$, or $s_v^{\gamma_{i}} = \bot$. Thus, $A_i$ is a candidate set.

\ \end{proof}

Then, when the timed move $(j,\delta)$ is labelled $(i,A_i)$, we say that it is of \motnouveau{color}~$i$.

\begin{lemma} \label{adv-lem-colorunique}
For any integer $j>0$, and any move $\delta \in t_j$, $(j,\delta)$ has a unique color.
\end{lemma}

\begin{proof}
Let $\delta= (u,r)$ be a move of $t_j$.  Recall that it means that $u$ executes  rule $r$ in $\gamma_{j-1} \xrightarrow{t_{j}} \gamma_{j} $.

Recall also that $(j,\delta)$ having color $i$ means by definition having label $(i,A_i)$. We must then prove that there exists an unique $ i\leq j$ such that $(j,\delta)$ has label $(i,A_i)$. We split cases according to the nature of rule $r$:
\begin{itemize}

\item If $r=\textbf{Candidacy}$, the guard of rule \textbf{Candidacy} guarantees that $s_u^{\gamma_{j-1}} = \bot$ and $s_u^{\gamma_{j}} = \top$, and then $u\in A_j$ by definition of $A_j$.
Thus $(j,\delta)$ has label $(j,A_j)$, \emph{i.e.} color $j$.
Moreover, by definition of the label of a \textbf{Candidacy} move, $(j,\delta)$ can only have a label that have $j$ as left-hand side. 
Then, by definition of the colors, it cannot have another color than $j$.

\item If $r=\textbf{Withdrawal?}$, we have $s_u^{\gamma_{j-1}} = \top$ from the guard of the \textbf{Withdrawal?} rule.
Then we look for the first index $\ell$ such that $u$ have continuously $s$-value $\top$ until the $j$-th transition: $$\ell = \min \quickset{i | i \le j \wedge \forall k \in \segment{i}{j-1}, s_u^{\gamma_k} = \top}$$
There we split cases according to the value of $\ell$:
\begin{itemize}
\item If $\ell = 0$, we have $u \in A_0$ by definition of $A_0$.
\item Otherwise, $\ell >0$ 
and then by definition of $\ell$ we have that $s_u^{\gamma_{\ell-1}} = \bot$. Thus $u \in A_{\ell}$, since $s_u^{\gamma_{\ell}}=\top$ by definition of $\ell$.
\end{itemize}
In both cases, since the state of $u$ is the same from the configuration of index $\ell$ to the configuration of index ${j-1}$, the timed move $(j,\delta)$ is labelled $(\ell,A_{\ell})$ and thus is of color $\ell$.

Now that we've proved that every timed move has a color, it remains to prove unicity.
Suppose by contradiction that $(j,\delta)$ is also of color $\ell'\not = \ell$, \emph{i.e.} of label $(\ell',A_{\ell'})$.
Having such label means that the $s$-value of $u$ remains $\top$ in configurations of index $\ell'$ to $j-1$.
Recall that $\ell$ is defined as the minimum index having that property, thus we must have $\ell < \ell'$. 
But then we have $\ell \leq \ell'-1 \leq j-1$, thus $s_u^{\gamma_{\ell'-1}} = \top$ by definition of $\ell$.
This means that $u \not\in A_{\ell'}$ by definition of $A_{\ell'}$, contradiction with $(j,\delta)$ being labelled $(\ell',A_{\ell'})$.
\end{itemize}
\ \end{proof}

We also extend the notion of label and color to \motnouveau{possible moves}. If $\gamma_i$ has $\delta$ as  possible move, then $(i+1,\delta)$ is of a certain label (resp. color) when it would have been of that label (resp. color) if $\delta$ was in $t_{i+1}$, every previous configurations and sets of moves being unchanged. Note that the above lemma is still true for this extended notion of color, as the proof only uses previous configurations.

As the definition of the label of a \textbf{Withdrawal?} move only depend on ``past'' moves, we can naturally extend that notion to \textbf{Withdrawal?} \motnouveau{possible moves}: if $\delta$ is a possible move of $\gamma_i$, then $(i+1,\delta)$ is of a certain label (resp. color) when it would have been of that label (resp. color) if it was in $t_{i+1}$, every previous configurations and sets of moves being unchanged.

The lemma that we will prove next tells us that when a label dies out, it remains dead forever: when there is no possible move that would have that label in a given configuration, no further configuration will have possible moves that would have that label. That will help us later count the number of timed move with a given label.


\begin{lemma} \label{adv-lem-colorend}
Let $i \leq j$, $A$ be a candidate set of configuration $\gamma_i$, and $u\in A$. If no possible move $\delta$ 
on $u$ in $\gamma_j$ 
is such that $(j+1,\delta)$ has label $(i,A)$, then any configuration $\gamma_k$ with $k > j$ has no possible move $\delta'$ for $u$ in $\gamma_k$ 
such that $(k+1,\delta')$ has label $(i,A)$.
\end{lemma}

\begin{proof}
Under the assumption of the lemma,  node  $u$ is in a candidate set $A$ of $\gamma_i$. Thus we have $s_u^{\gamma_{i}}=\top$. We are then in either of those cases:
\begin{itemize}
\item $\exists \ell \in \segment{i+1}{j}$, $s_u^{\gamma_{\ell}}=\bot$ ($u$ has changed state between the $i$-th and the $j$-th transition), and there can't be any move $\delta$ on $u$ such that $(k+1,\delta)$ would be labelled $(i,A)$ by definition of the labels.
\item $\forall \ell \in \segment{i+1}{j}$, $s_u^{\gamma_{\ell}}=\top$ ($u$ has not changed state between the $i$-th and the $j$-th transition), which menans that $s_u^{\gamma_j} = \top$. Since by hypothesis no possible move $\delta$ on $u$ in $\gamma_j$ is such that $(j+1,\delta)$ has label $(i,A)$, \textbf{Withdrawal?} is not enabled on $u$, which means that $u \in \beta(\gamma_j)$. Thus, using Lemma~\ref{adv-lem-betagrows}, $u \in \beta(\gamma_k)$, and no rule is enabled on $u$ in $\gamma_k$.
\end{itemize}
\
\end{proof}


Now, we focus on the  set $C_i$ of all the colors of possible \textbf{Withdrawal?} moves in configuration $\gamma_i$.
Let $Q_j^{i}$ be the set of nodes that can perform a \textbf{Withdrawal?} move of color $j$ in $\gamma_i$.

\begin{lemma} \label{adv-lem-114}
Given a configuration $\gamma_{i-1}$ and a transition $t_i$ we have:
\begin{enumerate}
\item If $t_i$ contains no \textbf{Candidacy} move, then $\sum_{j\in C_i} |Q_j^{i}| \leq \sum_{j\in C_{i-1}} |Q_j^{i-1}|$ and there is a probability at least $\frac{1}{2}$ that either:

\begin{itemize}
\item $\exists j \in C_{i-1}$ such that  $j \not\in C_{i}$.
\item $\sum_{j\in C_i} |Q_j^{i}| < \sum_{j\in C_{i-1}} |Q_j^{i-1}|$
\end{itemize}

\item If $t_i$ contains a \textbf{Candidacy} move, then either of those is true:
\begin{itemize}
\item $|\beta(\gamma_i)| > |\beta(\gamma_{i-1})|$ 
\item $\exists j \in C_{i-1}$ such that  $j \not\in C_{i}$.
\item $C_{i-1} \subseteq C_{i}$ and $|C_i| > |C_{i-1}|$.

\end{itemize}
\end{enumerate}
\end{lemma}

\begin{proof}
%
%
%
\begin{enumerate}
\item In order to prove point~1, let's suppose that $t_i$ does not contain any \textbf{Candidacy} move.

Let $j$ be a color in $C_{i-1}$.  If $u \in Q_j^i$, then $u$ can execute   \textbf{Withdrawal?}  rule in configuration $\gamma_{i}$, and the timed move $(i,(u,\textbf{Withdrawal?}))$ has color $j$.  Since $j<  i$,  by contraposition, Lemma~\ref{adv-lem-colorend} implies that 
$u$ can execute   \textbf{Withdrawal?}  rule in configuration $\gamma_{i-1}$, and the timed move $(i-1,(u,\textbf{Withdrawal?}))$ has color $j$.  
Thus,  $Q_j^i \subseteq Q_j^{i-1}$. 
Since only \textbf{Candidacy} move can create new potential color for moves, we have $C_i \subseteq C_{i-1}$ and thus $\sum_{j\in C_i} |Q_j^{i}| \leq \sum_{j\in C_{i-1}} |Q_j^{i-1}|$.
Moreover, since there is no \textbf{Candidacy} move, there must be at least a \textbf{Withdrawal?} move in $t_i$.
Then with probability $\frac{1}{2}$, $s_u^{\gamma_i}=\bot$. There are then two possibilities:
\begin{itemize}
\item If $C_{i-1} \subseteq C_{i}$, $s_u^{\gamma_i}=\bot$ implies that $|Q_k^i| < |Q_k^{i-1}|$ (since $u \notin  Q_k^i$). Thus, $\sum_{j\in C_i} |Q_j^{i}| < \sum_{j\in C_{i-1}} |Q_j^{i-1}|$ thus the inequality is then strict.
\item If not, by definition of relation $\subseteq$, we have $\exists j \in C_{i-1}$ such that  $j \not\in C_{i}$.
\end{itemize}

\item In order to prove point~2, let's suppose that $t_i$ contains a \textbf{Candidacy} move on some node $u$.
\begin{itemize}
\item If there is no possible move for $u$ in $\gamma_i$, then $u \in \beta(\gamma_i)$. Since $s_u^{\gamma_{i-1}} = \bot$ we have $u \not\in \beta(\gamma_i)$. Thus, $\beta(\gamma_{i-1} \uplus \quickset{u}) \subset \beta(\gamma{i})$ and thus $|\beta(\gamma_{i-1})|<|\beta(\gamma{i})|$
\item If not, there is a possible move for $u$ in $\gamma_i$, which is $(u,\textbf{Withdrawal?})$ of color $i$, thus $i \in C_i$
\begin{itemize}
\item If $C_{i-1} \subseteq C_{i}$, since $i \leq i-1$ we have $i\not\in C_{i-1}$ and thus $C_{i-1}\uplus \quickset{i} \subseteq C_{i}$, which gives $|C_{i-1}| < |C_{i}|$.
\item If not, by definition of relation $\subseteq$, we have $\exists j \in C_{i-1}$ such that  $j \not\in C_{i}$.
\end{itemize}
\end{itemize}
\end{enumerate}
\ \end{proof}

\begin{lemma} \label{adv-lem-115}
From any configuration $\gamma_i$ such that $C_i \not= \emptyset$, the expected number of transitions to have a color disappearing is finite. 
\end{lemma}

\begin{proof}
From Lemma~\ref{adv-lem-colorunique}, we have that for any $k$, $0 \leq |C_k| \leq n$ and $0 \leq 2|C_k| \leq \sum_{j\in C_k} |Q_j^{k}| \leq n$. Thus, there is at most $(n+1)^2$ possible values for $(|C_k|, -\sum_{j\in C_k} |Q_j^{k}|)$. We know from Lemma~\ref{adv-lem-114} that either a color disappear, $\beta$ grows (which cannot happen more than $n$ times), or this value strictly increases (for the lexicographical order) with probability at least $\frac{1}{2}$, at every transition. Thus, after an expected number of transition of at most $ 2(n+1)^2$ either a color disappeared, $\beta$ grew, or we ended in a configuration $\gamma_{k_0}$ where $(|C_{k_0}|, -\sum_{j\in C_{k_0}} |Q_j^{k_0}|) = (n,-2n)$. From $\gamma_{k_0}$, using Lemma~\ref{adv-lem-114}, the expected number of transitions for a color to disappear or for $\beta$ to grow is at most $2$ since $(|C_{k_0}|, -\sum_{j\in C_{k_0}} |Q_j^{k_0}|)$ cannot increase anymore.

Thus, in any case, the expected number of transitions for either a color to disappear or $\beta$ to grow is at most $2(n^2+2n+2)$. Since the size of $\beta$ cannot be greater than $n$, the expected number of transitions for a color to disappear is at most $n \times 2(n^2+2n+2)$.
\ \end{proof}

We need the next lemma in order to use proof by induction, as nodes of a given candidate set $A$ in a given configuration may take value $\bot$ in a further configuration, $A$ would not then be a candidate set of that configuration. We prove a kind of hereditary property in order to make our induction work.

\begin{lemma} \label{adv-lem-inductionworks}
Let $i \leq j$. Let $A$ be a candidate set of $\gamma_i$, and \\$A' = \quickset{u \in A | \forall k \in \segment{i}{j}, s_u^{\gamma_k} = \top }$. Then:
\begin{enumerate}

\item $A'$ is a candidate set of $\gamma_j$;
\item Every move labelled $(j,A')$ is also labelled $(i,A)$.
\end{enumerate}
\end{lemma}

\begin{proof}
Let's first prove Point~1 by contradiction.
Suppose that $A'$ is not a candidate set of $\gamma_j$. It implies that there exist two nodes  $u \in A'$ and $v \in N(u)$ such that $s_v^{\gamma_j} = \top$ and $v \not\in A'$. There are two possible cases:
\begin{itemize}
\item If $v \not\in A$, then $s_v^{\gamma_i}=\bot$ since $A$ is a candidate set of configuration $\gamma_i$, which gives us that $i \in \quickset{x \in \segment{i}{j} | s_v^{\gamma_x} = \bot}$. 
\item If $v \in A$, then $s_v^{\gamma_i}=\top$, and since $v \not\in A'$ by definition of $A'$ the set $\quickset{x \in \segment{i}{j} | s_v^{\gamma_x} = \bot}$ is nonempty. 
\end{itemize}

\noindent In both cases $\quickset{x \in \segment{i}{j} | s_v^{\gamma_x} {=} \bot}$ is nonempty, and thus $\ell {=} max(\quickset{x \in \segment{i}{j} | s_v^{\gamma_x} {=} \bot})$ is well-defined, and $\ell \not= j$ since $s_v^{\gamma_j} = \top$ by hypothesis.
Then, $s_v^{\gamma_{\ell}} = \bot$ and $s_v^{\gamma_{\ell+1}} = \top$, thus $(v,\textbf{Candidacy})\in t_{\ell+1}$. But we know that $s_v^{\gamma_{\ell}} = \top$ (from $u \in A$) and that $u$ and $v$ are neighbours, so the guard of \textbf{Candidacy} cannot be true on node $v$ in $\gamma_{\ell}$, contradiction with $(v,\textbf{Candidacy})\in t_{\ell+1}$.

Now let's prove Point~2. We consider a timed move $(k,(u,r))$ labelled $(j,A')$
with $r$ a rule $u$ a node, and $k \in \Nat$. By definition of the labels we must have $j \leq k$ and $u \in A'$.

\begin{itemize}
\item If $r=\textbf{Candidacy}$, $s_u^{\gamma_{j-1}}$ is  $ \bot$,  and the definition of the label of a timed move induces  that $i=j$. Then, we have $A'=  \quickset{u \in A | \forall k \in \segment{i}{i}, s_u^{\gamma_k} = \top }= A$ thus $(j,A')=(i,A)$ and the timed move $(k,(u,r))$ is then labelled $(i,A)$.
\item Else, we have $r=\textbf{Withdrawal?}$.
By definition of the label of a \textbf{Withdrawal?} timed move, we have that $\forall \ell \in \segment{j}{k-1}, s_u^{\gamma_{\ell}} = \top$. Then by definition of $A'$, $\forall \ell \in \segment{i}{j}, s_u^{\gamma_{\ell}} = \top$. Thus, $\forall \ell \in \segment{i}{k-1}, s_u^{\gamma_{\ell}} = \top$, thus $(k,(u,r))$ is a timed move labelled $(i,A)$.
\end{itemize}
\ \end{proof}


The following two lemmas allow us to make links between different configurations that share properties, in order to use induction reasoning afterward. 


%
%


\begin{lemma} \label{adv-lem-111}
Let $\gamma$ and $\gamma'$ be two configurations and $u$ a node.
If $u$ and each of its neighbours have the same $s$-value in $\gamma$ and $\gamma'$ then the possible moves on $u$ are the same in $\gamma$ and $\gamma'$.
\end{lemma}

\begin{proof}
It is enough to notice that guards may only check the value of attributes of the node itself and of its neighbours .
\ \end{proof}

We denote by $\Xi_{\gamma,t}(u)$  the random variable describing the $s$-value  of $u$  after the transition $t$ from configuration $\gamma$.

\begin{lemma} \label{adv-lem-samerulesamevalue}
Let $\gamma$ and $\gamma'$ be two configurations and $u$ a node.
Let $t$ and $t'$ be two valid sets of moves of respectively $\gamma$ and $\gamma'$.

If $\delta$ is a move on $u$ such that $\delta \in t \cap t'$, we have $\Xi_{\gamma,t}(u)=\Xi_{\gamma',t'}(u)$ (\emph{i.e} $\forall x \in \quickset{\top,\bot},\Pr(\Xi_{\gamma,t}(u)= x) =\Pr( \Xi_{\gamma',t'}(u)=x)  $)
\end{lemma}

\begin{proof}
As the effects of any rule of our algorithm do not depend on the neighbours (provided they were enabled), the random variable describing the value of $u$ is the same in both cases.
\ \end{proof}

Recall that  $V(t)$ is the set of nodes who executes a move in the set $t$.

\begin{corollary} \label{adv-lem-candidatesetindependent}
Let $\gamma$ and $\gamma'$ be two configurations.
Let $A$ be a candidate set of both configurations.

We have that:
\begin{itemize}
\item The possible moves on nodes in $A$ are the same in $\gamma$ and $\gamma'$.
\item If $t$ and $t'$ be two valid sets of moves of respectively $\gamma$ and $\gamma'$ such that $V(t)\cap A = V(t') \cap A$, then $\forall u \in A$, $\Xi_{\gamma,t}(u)=\Xi_{\gamma',t'}(u)$.
\end{itemize}
\end{corollary}

%
%
%


The next lemma is a technical one. We prove that if families probabilities have a certain set of properties, then then we have a lower bound for the value of those probabilities. The goal is to apply this lemma to probabilities that emerge from our induction afterward.

\begin{lemma} \label{adv-lem-proba}
Let $(p_y)_{y\in \mathbb{N}}$ and $(p_{y,z})_{(y,z)\in \mathbb{N}^2}$ be two families of probabilities, defined for $y,z$ integers such that $0 \leq z \leq y$, which have the following properties: 
\begin{enumerate}
\item $p_0 = 0$;
\item $p_1 = 1$;
\item $\displaystyle p_{y,z} \geq \frac{1}{2^z}\sum_{\ell=0}^{z} \binom{z}{\ell} p_{y-\ell}$;
\item $p_y = min(\quickset{p_{y,z} | 1 \leq z \leq y}$ for $y\geq 2$;
\end{enumerate}

Then we have  the following property: $\forall y \geq 1, p_y \geq \frac{2}{3}$.
\end{lemma}

\begin{proof}
We prove this by induction on $y$.

For $y=1$, $p_y = 1 \geq \frac{2}{3}$. Assume now that $y\geq 2$ and that the result holds for all values in $\segment{1}{y-1}$.

For any $z \in \segment{1}{y}$, by rewriting Point~3. of the definition, we have:
$$p_{y,z}  \geq    \frac{1}{2^z} p_{y} + \frac{1}{2^z}\sum_{\ell=1}^{z} \binom{z}{\ell} p_{y-\ell} $$

Since $p_y = min(\quickset{p_{y,z} | z \in \segment{1}{y}})$, there exists $z_0 \in \segment{1}{y}$ such that $p_{y,z_0}=p_y$, which leads to, by substitution of $p_y$ in the previous equation:

$$p_{y,z_{0}}  \geq    \frac{1}{2^{z_{0}}} p_{y,z_{0}}  + \frac{1}{2^z_{0}}\sum_{\ell=1}^{z_{0}} \binom{z_{0}}{\ell} p_{y-\ell} $$

$$\frac{2^{z_0}-1}{2^{z_0}}  p_{y,z_0} \geq \frac{1}{2^{z_0}}\sum_{\ell=1}^{{z_0}} \binom{z_0}{\ell} p_{y-\ell}$$

\begin{equation} \label{eq1} p_y=p_{y,z_0} \geq \frac{1}{2^{z_0} -1}\sum_{\ell=1}^{{z_0}} \binom{z_0}{\ell} p_{y-\ell}
\end{equation}

\begin{itemize}
\item If $z_0 = y$, then we have $p_{y-z_0} = p_{z_0-z_0} = p_0 = 0$ which gives, by rewriting in (Eq.~\refeq{eq1}):

$$p_y=p_{y,y} \geq \frac{1}{2^{y} -1}\sum_{\ell=1}^{{y}-1} \binom{y}{\ell} p_{y-\ell}$$
We can then write by getting the first element of the sum out (taking the convention that $\sum_{\ell=1}^{0} s(\ell)$ sums to $0$):
$$p_y \geq \frac{1}{2^{y} -1}\sum_{\ell=1}^{{y}-1} \binom{y}{\ell} p_{y-\ell} = \frac{1}{2^{y} -1} \left(y \cdot p_{1} + \sum_{\ell=1}^{{y}-2} \binom{y}{\ell} p_{y-\ell} \right) $$
Then, since $1 \leq y-\ell \leq y-1$ for $\ell \in \segment{1}{y-2}$, by applying the induction hypothesis (which gives $p_{y-\ell}\geq \frac{2}{3}$ for $\ell \in \segment{1}{y-2}$) we obtain:
$$p_y \geq \frac{1}{2^{y} -1} \left(y + \frac{2}{3}\sum_{\ell=1}^{{y}-2} \binom{y}{\ell} \right)$$

Using the fact that  $\displaystyle \sum_{\ell=0}^{{y}} \binom{y}{\ell} = 2^y$, we can rewrite the right hand side to get:
$$p_y \geq \frac{1}{2^{y} -1} \left(y\cdot p_{1} + \frac{2}{3}\left(2^y-y-2 \right) \right)$$

$$p_y \geq \frac{2}{3} + \frac{1}{2^{y} -1} \cdot \frac{y-2}{3} \geq \frac{2}{3}$$
\item Now, if $z_0 \not= y$, we have $z_0 < y$. We can then apply the induction hypothesis on the $p_\ell$ for $\ell \in \segment{1}{z_0}$ to obtain, by rewriting in (Eq.~\refeq{eq1}):

$$p_y \geq \frac{1}{2^{z_0} -1}\sum_{\ell=1}^{z_0} \binom{z_0}{\ell}\frac{2}{3}$$

Using the fact that  $\displaystyle \sum_{\ell=0}^{{z_0}} \binom{z_0}{\ell} = 2^{z_0}$, we get:
$$p_y \geq \frac{1}{2^{z_0} -1}\frac{2}{3}(2^{z_0}-1)= \frac{2}{3}$$
\end{itemize} 
\ \end{proof}

\begin{remark}
Using Lemma~\ref{adv-lem-colorend} one can note that for a given label $(i,A)$, if there exists $j \geq i$ an index such that there is no possible timed move labelled $(i,A)$ in $\gamma_j$, it will still be the case in the future.
\end{remark}

When a color disappears, it's time to settle its accounts: did it make the independent set grow or not during its lifetime? In order to answer probabilistically to this question, we need to answer the more general question about labels.


\newcommand{\lb}[1]{{q_{#1}}}


\newcommand{\exec}{\mathcal{E}}

Let $\exec=\gamma_0 \xrightarrow{t_1} \gamma_1 \dotsm \gamma_{i-1} \xrightarrow{t_i} \gamma_{i}$ be a partial execution, and $A$ a candidate set of the last configuration of $\exec$. Then we denote by $\lb{\exec,A}$ the minimum of the probability (whatever the daemon's strategy) that in the completed execution $\exec'$ there exists $k \in \Nat$ and $u\in A$ such that:
\begin{itemize}
\item There is no possible move labelled $(i,A)$ in $\gamma_{i+k}$, and $k$ is the smallest index greater than $i$ for which it is true.
\item All moves on $u$ in $t_{i+1}, \dotsm ,t_{i+k}$ are labelled $(i,A)$;
\item $u \in \beta(\gamma_{i+k})$.
\end{itemize}

\begin{notation} Whe $\exec$ is a finite execution, we will write $\last(\exec)$ the last configuration of $\exec$. Note that every partial execution is finite.
\end{notation}

\begin{lemma} \label{adv-lem-qia-qa}
The minimum $\lb{\exec,A}$ do not depend on the execution, but only on $A$. 
\end{lemma}

\begin{proof}
Let us first remark that Lemma~\ref{adv-lem-colorend} garantees that such $k \in \Nat$ exists with probability $1$.

Now, to prove this, we define $\lb{\exec,A,k}$ the same as $\lb{\exec,A}$ but where the said first index such that there is no possible move labelled $(i,A)$ is exactly $k$.
By slicing the sample space for every possible $k$, we have $\displaystyle \lb{\exec,A} = \sum_{k=0}^{+\infty}\lb{\exec,A,k}$ (we can sum since different $k$ means disjoint events). It is then enough to prove that $\lb{\exec,A,k}$ can be expressed as a function of $A$ and $k$ to prove our lemma.

Let's prove that by induction on $k \geq 0$ that for any partial execution $\exec$
 and $A$ a candidate set of $\last(\exec)$, $\lb{\exec,A,k}$ can be expressed as a function on $A$ and $k$.

For $k=0$, we consider two cases.
\begin{itemize}
\item If there is $u\in A$ with no neighbour in $A$, we have $u \in \beta(\last(\exec))$ thus $\lb{\exec,A,k}=1$.
\item Otherwise, for any $u\in A$ there is a $v\in A$ neighbour of $u$, and since $s_v^{\last(\exec)}=\top$, we have $u\not\in \beta(\last(\exec))$. Thus, $\lb{\exec,A,k}=0$.
\end{itemize}
In both cases, it does not depend on $\exec$.


Now, we suppose that $k\geq 1$. We slice again the sample space using all the possible valid set of moves by which the execution could continue from the last configuration of $\exec$. We define $q_{\exec,A,k|t}$ which is the same probability than $\lb{\exec,A,k}$ but knowing that the first transition that follow $\exec$ use the valid set of moves $t$.
As the daemon may choose the $t$ that achieves the minimum, we can write:
\begin{equation}\label{eq:q0}
\lb{\exec,A,k}= \min_{t \text{ valid in} \last(\exec)}q_{\exec,A,k,|t}
\end{equation}

We want to write $q_{\exec,A,k|t}$ as a sum where each term represent the case where the configuration after the transition using $t$ from the last configuration of $\exec$ is a given $\gamma$. To do this, we define 
$A'_{\gamma}=\quickset{u \in A | s_u^{\gamma} = \top }$.
As $A'_{\gamma} \subset A$ and $A$ is a candidate set of $\last(\exec)$, every member of $A'_{\gamma}$ has also value $\top$ in $\last(\exec)$. Thus, $A'_{\gamma} = \quickset{u \in A | s_u^{\last(\exec)} = s_u^{\gamma} = \top }$.
Considering the partial execution $\exec \xrightarrow{t} \gamma$, we can then use Lemma~\ref{adv-lem-inductionworks} to deduce that:
\begin{itemize}
\item $A'_{\gamma}$  is a candidate set of $\gamma$
\item Every further move labelled $(i+1,A'_{\gamma})$ will also be labelled $(i,A)$ (where $i$ is the length of $\exec$, \emph{i.e.} its number of transition)
\end{itemize}
$q_{\exec,A,k|t}$ can then be written in the following form, by decomposing according every possible resulting configuration after performing the transition using the set of moves $t$:
 \begin{equation}\label{eq:q1}
 \lb{\exec,A,k|t} = \sum_{\gamma} \Pr(\Xi_{\last(\exec), t} = \gamma)\lb{\exec\xrightarrow{t} \gamma,A'_{\gamma},k-1} \end{equation}

Recall that $\Xi_{\last(\exec),t}$ is the random variable describing the configuration after the transition $t$ from configuration $\last(\exec)$.


By induction hypothesis, $\lb{\exec\xrightarrow{t} \gamma,A'_{\gamma},k-1}$ does not depend on the execution. For sake of simplification, we will now denote this value $\lb{.,A'_{\gamma},k-1}$, omiting the execution that it does not depend upon. 
We have then, by regrouping in Eq.~\eqref{eq:q1} every $\gamma$ which give the same $A'_{\gamma}$:

$$ q_{\exec,A,k|t} = \sum_{\gamma} \Pr(\Xi_{\gamma_i, t} = \gamma)q_{.,A'_{\gamma},k-1}= \sum_{B\subset A} \left( q_{.,B,k-1} \sum_{\quickset{\gamma | A'_{\gamma}=B}} \Pr(\Xi_{\gamma_i, t} = \gamma) \right)$$

And then by noting that the internal sum is the probability that $\gamma$ has value $\top$ on nodes of $B$ and $\bot$ on other nodes of $A$ we can write:
\begin{equation}\label{eq:qq2}
q_{\exec,A,k|t} =\sum_{B\subset A}  q_{.,B,k-1} \Pr(\forall u\in A, \Xi_{\last(\exec), t}(u) = \top \text{ if } u\in B, \text{else } \bot) 
\end{equation}

From Corollary~\ref{adv-lem-candidatesetindependent} $\Pr(\forall u\in A, \Xi_{\last(\exec), t}(u) = \top \text{ if } u\in B, \text{else } \bot)$ is a function of $A$ and $t$. Then from Eq.~\eqref{eq:qq2}, we have $q_{\exec,A,k|t}$ as a function of $A$, $k$ and $t$, let us write $\alpha(A,k,t)$ that quantity.

Then we can rewrite $q_{\exec,A,k| t}$ in Eq.~\eqref{eq:q0}:
$$\lb{\exec,A,k} = \min_{t \text{ valid in } \last(\exec)} \alpha(A,k,t)$$

We can then split the minimum for each possible set of activated nodes in the transition, which gives:

$$\lb{\exec,A,k} = \min_{A' \subset A} \left( \min_{t \text{ valid in } \last(\exec), V(t)\cap A = A'}\alpha(A,k,t) \right)$$

From Lemma~\ref{adv-lem-111} the possible moves on nodes of $A$ in $\last(\exec)$ is a function of $A$, the existence for a given $A'$ of a $t$ valid in $\last(\exec)$ such that $V(t)\cap A = A'$ is then function of $A$ and $A'$, $f(A,A')$. Thus, $\displaystyle \lb{\exec,A,k} {=} \min_{A' \subset A} \left( \min_{f(A,A')}\alpha(A,k,V(t)\cap A) \right)$ is a function of $A$ and $k$.

Thus, by induction, $\lb{\exec,A,k}$ can be expressed as a function of $A$ and $k$. Then ${\displaystyle \lb{\exec,A} = \sum_{k=0}^{+\infty}\lb{\exec,A,k}}$ can be expressed as a function of $A$.
\ \end{proof}

Since $\lb{\exec,A}$ does not depend on the actual execution, we will write it $q_{A}$ from now on.

In the same way as we have defined $\lb{\exec,A}$, we want to define $q_{\exec,A|A'}$ as the same thing as $\lb{\exec,A}$ but knowing that the first set of nodes of $A$ to be activated to be is exactly $A'$. However, it would not be well-defined when it is impossible for $A'$ to be the first set of nodes to be activated in $A$. It happens when $A'$ is empty, or when a node of $A'$ is not activable in $\last(\exec)$. $q_{\exec,A|A'}$ will then be defined as:
\begin{itemize}
\item The same thing as $\lb{\exec,A}$ but knowing that the first set of nodes of $A$ to be activated to be is exactly $A'$ when every node of $A'$ is activable in $\last(\exec)$.
\item In the case where $A'$ is empty, or when there is a node in $A'$ that cannot be activated in the first transition that activate a node of $A$, we take the convention $\lb{\exec,A|A'}=1$. As we will manipulate those values with minimum functions, it is handy to have those particular cases being neutral element for that operation.
\end{itemize}

\begin{lemma} \label{adv-lem-qiaa-qaa}
$q_{\exec,A|A'}$ can be expressed as a function of $A$ and $A'$.
\end{lemma}

\begin{proof}
By the same set of arguments as in Lemma~\ref{adv-lem-qia-qa}.
\ \end{proof}

As $q_{\exec,A|A'}$ does not depend on the actual configuration, and we will write it $q_{A|A'}$ from now on.

\begin{remark}
We note by $\gamma_A$ the configuration where every node of $A$ has value $\top$ and every other node has value $\bot$.

Note that $q_A$ is defined for every possible $A \in \mathcal{P}(V)$ ($\mathcal{P}(V)$ being the powerset of $V$): when we take the execution reduced to a single configuration $\gamma_A$, $A$ is a candidate set of the last configuration of the said execution, and thus $q_{\gamma_A,A}=q_A$ is defined.

In the same way, $q_{A|A'}$ is defined for every possibles $A,A'$ where $A \in \mathcal{P}(V)$ and $A' \subset A$.
\end{remark}

\begin{lemma} \label{adv-lem-induction}
Given a partial execution $\exec$ 
and $A$  a non-empty candidate set of 
$\last(\exec)$, we have $\lb{\exec,A} \geq \frac{2}{3}$.
\end{lemma}

\begin{proof}
We want to denote by $r_y$ the minimum of the $q_A$ when $A$ has the constraint to be of size $y$. In the same way, we want to define $r_{y,z}$ as the minimum of the $q_{A,A'}$ when $A$ is of size $y$ and $A'$ of size $z$.
We write for $(y,z) \in \Nat^2$: 
\begin{itemize}
\item $r_y = \min (\quickset{q_{A}|A \in \mathcal{P}(V) \wedge |A|=y})$;
\item $r_{y,z} = \min(\quickset{q_{A|A'} | A \in \mathcal{P}(V)  \wedge |A|=y \wedge A' \subseteq A \wedge |A'| = z}$
\end{itemize}

We will prove that $(r_y)_{y\in \Nat}$ and $(r_{y,z})_{(y,z)\in \Nat^2}$ verify the conditions of Lemma~\ref{adv-lem-proba}.

\begin{itemize}
\item If $A \subset V$ is of size $0$, we have $A=\emptyset $, and $q_{\gamma_{\emptyset},\emptyset} = q_{\emptyset} = 0$, thus $r_0 = 0$. Also, by definition, $q_{\gamma_{\emptyset},\emptyset|\emptyset} = 0$ thus $r_{0,0}=0 \geq r_0$.

\item If $A \subset V$ is of size $1$, let's consider the execution reduced to a single configuration $\gamma_A$. Since $A$ is of size $1$, there is no possible \textbf{Withdrawal} move thus no further possible move labelled $(0,A)$, and the only node $u$ in $A$ is in $\beta(\gamma_A)$. Thus $q_A =1$ when $|A|=1$, for any choice of $A$. Thus $r_1 = 1$.

There, we have $r_{1,0}=q_{A,\emptyset} = 1$ by definition and $r_{1,1} = q_{A,A} = 1$ by the same argument than for $q_A$.

\item Now consider $A \subset V$ is of size $2$ or greater.

If there exists a node $u$ in $A$ not activable  in $\gamma_A$, then $u\in \beta{\gamma_A}$ and then $q_{\gamma_A,A}= q_A = 1$.

Now, let us assume that every node of $A$ is activable in $\gamma_A$.

The daemon's strategy that would give the minimal probability $q_{A}$ begins with the activation of a certain subset $A'$ of nodes from $A$. As Lemma~\ref{adv-lem-115} garantees that every color disappear with probability $1$, $A'$ cannot be empty. We can then write:

\begin{equation} \label{eq:minqa} 
q_{A} = \min\quickset{q_{A|A'} | A' \subseteq A \wedge A'\not=\emptyset}
\end{equation}

Thus, we have, taking the minimum for every possible set of the same size as $A$ in Eq.~\refeq{eq:minqa}:\\[0,3em]

\hspace*{1cm}  ${\min}\quickset{q_{B} | B {\in} \mathcal{P}(V) \wedge |B| {=} |A|} =$\\[5pt]
\hspace*{2cm}  $ {\min} \quickset{{\min}\quickset{q_{B|B'} | B' {\subseteq} B \wedge B'{\not=}\emptyset}| B \in \mathcal{P}(V) \wedge |B| {=} |A|}$\\[5pt]

The left hand side is exactly $r_{|A|}$, and by rewriting the right hand side we have:

$$r_{|A|} = \min\quickset{q_{B|B'} | B' \subseteq B \wedge B'\not=\emptyset \wedge B \in \mathcal{P}(V) \wedge |B| = |A|}$$

We then rewrite the right hand side by slicing the min according the different possible size for $B'$:

$$r_{|A|} {=} \min\quickset{ min \quickset{q_{B|B'} | B \in \mathcal{P}(V) \wedge |B| {=} |A| \wedge B' \subseteq B \wedge |B'|{=} z} | 1 \leq z \leq |A|}$$

We can then note that the definition of $r_{|A|,z}$ appears in the right hand side and rewrite the equation in the following form: 
 $$r_{|A|} =  \min\quickset{r_{|A|,z}| 1 \leq z \leq |A|}$$

Thus for $y\geq 2$ we have $r_{y} =  \min\quickset{r_{y,z}| 1 \leq z \leq y}$
%
%
%

\item To apply Lemma~\ref{adv-lem-proba} it remains to prove that  we have $\displaystyle r_{y,z} \geq \frac{1}{2^{|A'|}}\sum_{\ell=0}^{z} r_{y-\ell}$ for $y\geq 2$  and $z\leq y$.

To do this, let's consider $A$ of size $y\geq 2$ and $A' \subset A$ of size $z \leq y$.

When $z=0$, it means that $A'=\emptyset$, and by definition we have $q_{A|A'}=q_{\gamma_A,A|A'} = 1$ for any value of $A$. Thus, $r_{y,0}=1 \geq r_y$.

If there exists a node $u$ in $A$ not activable in $\gamma_A$, then $u\in \beta_{\gamma_A}$, and then we have $q_{A,A'}=1$ (which is greater than any probability).

Suppose now that all nodes of $A$ are activable in $\gamma_A$ and that $z>0$. We take the partial execution reduced to a single configuration $\gamma_A$, and we consider all the executions where the first time nodes of $A$ are activated, it is exactly the nodes of $A'$. Corollary~\ref{adv-lem-candidatesetindependent} implies that regarding possible moves and and probabilities of evolution of the states of $A$ is the same in any two configurations where $A$ is candidate set and do not depend on the nodes activated outside $A$. As such, we may without loss of generality suppose that the first transition to activate nodes of $A$ activates exactly the nodes of $A'$. Also, Lemma~\ref{adv-lem-inductionworks} garantees that if no move are made on $A$, $A$ will remain a candidate set, and further moves on $A$ will still have the label $(0,A)$ \emph{i.e.} color $0$ as long are they are executed on nodes that didn't change their $s$-value. As such, we can make without loss of generality that the transition where $A'$ is activated is the first.

Every activated node -the nodes of $A'$- have then probability $\frac{1}{2}$ to change its $s$-value to $\bot$. The notation $A''$ will be used to denote the set of the nodes that succeded in changing their $s$-value, which gives the following:

$$q_{\gamma_A,A|A'} = \frac{1}{2^{|A'|}}\sum_{A''\subseteq A'} q_{\gamma_A\xrightarrow{t} \gamma_{(A\setminus A'')},(A\setminus A'')}$$

Lemma~\ref{adv-lem-qia-qa} gives $q_{\gamma_A,A|A'}{=}q_{A,A'}$, Lemma~\ref{adv-lem-qiaa-qaa} gives ${q_{\gamma_A\xrightarrow{t} \gamma_{(A\setminus A'')},(A\setminus A'')} {=} q_{(A\setminus A'')}}$, thus we can rewrite the above equation as:

$$q_{A|A'} = \frac{1}{2^{|A'|}}\sum_{A''\subseteq A'} q_{(A\setminus A'')}$$

By definition we have $q_{(A\setminus A'')}\geq r_{|A| - |A''|}$ (since $|A\setminus A''| = |A| - |A''|$), we can then rewrite in the above equation, and regroup the $r$ of same index:

$$q_{A|A'} \geq \frac{1}{2^{|A'|}}\sum_{A''\subseteq A'} r_{|A|- |A''|} = \frac{1}{2^{|A'|}}\sum_{\ell=0}^{|A'|}\binom{|A'|}{\ell} r_{|A|-\ell}$$

Thus, by taking the minimum over all potential $A$ and $A'$, the left hand side becomes $min(\quickset{q_{B|B'} | A \in \mathcal{P}(A) \wedge |B|=|A| \wedge B' \subset B \wedge |B'|=|A'|})$ which is by definition $r_{|A|,|A'|}$. By the same operation the right hand side remains the same as it only depends on the size of the said sets. Therefore we have:


$$r_{|A|,|A'|} \geq \frac{1}{2^{|A'|}}\sum_{\ell=0}^{|A'|}\binom{|A'|}{\ell} r_{|A|-\ell}$$

\emph{i.e.}

$$r_{y,z} \geq \frac{1}{2^{z}}\sum_{\ell=0}^{z}\binom{z}{\ell} r_{y-\ell}$$
\end{itemize}

Thus, $(r_y)_{y\in \Nat}$ and $(r_{y,z})_{(y,z)\in \Nat^2}$ verify the conditions of Lemma~\ref{adv-lem-proba} and thus $\forall y\geq 1, r_y \geq \frac{2}{3}$.

Since, by definition of $r$, $q_{A} \geq r_{|A|}$, we have then $q_{A} \geq \frac{2}{3}$ which concludes the proof.

\ \end{proof}

In the above Lemma, we made an assumption that implies that a certain candidate set is activated until there is no more possible move associated with that candidate set. Now, it remains to prove that it ever happens.
\begin{lemma} \label{adv-lem-116}
From any configuration $\gamma_i$, the expected number of moves to reach a stable configuration is finite.
\end{lemma}

\begin{proof}
From any configuration $\gamma_i$,   function $\beta$ can only grow (Lemma~\ref{adv-lem-betagrows}). From Lemma~\ref{adv-lem-115}, the expected number of moves for a color to disappear is at most $2(n^2+2n+2)$, and since $|C_i| \leq n$ the expected number of moves for a color to appear and then disappear is at most $(n+1)2(n^2+2n+2)$ (when the at most $n$ colors in $|C_i|$ have all disappeared, only new color can disappear). Then, from Lemma~\ref{adv-lem-induction}, the expected number of moves for $\beta$ to grow is at most $\frac{2}{3}(n+1)2(n^2+2n+2)$. Since $\beta$ cannot be greater than the set of nodes, the expected number of steps for $\beta$ to grow to a maximum size is at most $n\frac{2}{3}(n+1)2(n^2+2n+2) < +\infty$
\ \end{proof}

Now that we know that the computation converges, we can try to bound our expected number of steps of convergence more tightly.

\begin{lemma} \label{adv-lem-117}
If configuration $\gamma_i$ has $A$ as a non-empty candidate set, no more move labelled $(i,A)$ are possible  after  average less than $2|A|$ moves.
\end{lemma}

\begin{proof}
Every node has at least probability $\frac{1}{2}$ to change state every number of moves it is activated, leading to an expected number of activations before not being able to perform a move labelled $(i,A)$ of $2$ or less. The expected number of timed moves labelled $(i,A)$ is then on average less than $2|A|$.
\ \end{proof}


\begin{theorem}
From any configuration $\gamma$, the expected number of steps to reach a stable configuration is at most $3n^2$.
\end{theorem}

\begin{proof}

When reaching a stable configuration, every color has disappeared. From Lemma~\ref{adv-lem-116}, we reach a stable configuration after finite time. Every color that has been active in the computation has appeared then disappeared, and have added a new $\beta$ member with probability at least $\frac{2}{3}$ (see Lemma~\ref{adv-lem-induction} and since every color is a label), leading to an expected number of colors of at most $\frac{3}{2}n$ considering $\beta$ is smaller than $V$. Then from Lemma~\ref{adv-lem-117}, the expected number of steps to reach a stable configuration is at most $\frac{3}{2}n \times 2n = 3n^2$ since a candidate set cannot be greater than the total number of nodes.
\ \end{proof}

\newpage

\setcounter{theorem}{1}
\setcounter{algorithm}{1}
\bibliographystyle{plain}

\newpage
\appendix

\section*{Appendix}

\section*{Omitted proofs of Part~2}
 \setcounter{lemma}{-1}
\begin{lemma}
$\forall k \in \mathbb{N}, (1-\frac{1}{k+1})^{k} > e^{-1}$

\end{lemma}

\begin{proof}
 For $k=0$ it is true ($(1-\frac{1}{0+1})^{0} = 1 > \frac{1}{e}$).

Suppose now that $k \geq 1$.
A basic inequality about $\ln$ is that $\forall x\in \mathbb{R}^{+*}, \ln(x) \geq 1-x$, with equality only when $x=1$.
For $x=\frac{k}{k+1}$, it gives us $\ln \left(\frac{k}{k+1} \right) \geq 1 - \frac{k+1}{k}= -\frac{1}{k}$, with equality only when $\frac{k}{k+1} = 1$. But since $\frac{k}{k+1}$ cannot have value $1$ for any value of $k$, we have then $\ln \left(\frac{k}{k+1} \right) > -\frac{1}{k}$.

Then, by multiplying by $k$ on each side, we have $k \ln \left(\frac{k}{k+1} \right) > -1$, and thus, taking the exponential:
$$\left(1 - \frac{1}{k+1} \right)^k = e^{k \ln \left(\frac{k}{k+1} \right)} > e^{-1}$$
\ \end{proof}

 \setcounter{lemma}{1}
 
\begin{lemma}
Any reachable configuration from a degree-stabilized configuration is degree-stabilized.
\end{lemma}

\begin{proof}
No rule can change the $x$ value of a degree-stabilized non-Byzantine node.
\ \end{proof}

\begin{lemma}
From any configuration $\gamma$, the configuration $\gamma'$ after one round is degree-stabilized.
\end{lemma}

\begin{proof}
Let $u$ be a non-byzantine node.

If $x_{u}^{\gamma}=deg(u)$, no activation of rule can change that thus $x_{u}^{\gamma'}=deg(u)$.

If $x_{u}^{\gamma}\not = deg(u)$, then $u$ is activable in $\gamma$ and remains so until it is activated. Rule \textbf{Refresh} is then executed on $u$ in the first round, and since no rule can change the value of $x_u$ afterward we have $x_{u}^{\gamma'} = deg(u)$.
\ \end{proof}

\begin{lemma}

If $\gamma \rightarrow \gamma'$, $I_{\gamma} \subseteq I_{\gamma'}$.
\end{lemma}

\begin{proof}
Let's consider $u\in I_{\gamma}$. The only rules that may be enabled on $u$ in $\gamma$ is \textbf{Refresh} since \textbf{Candidacy?} can't be executed on $u$ because $s_u^{\gamma} = \top$, and \textbf{Withdrawal} can't be executed on $u$ because $\forall v\in N(u), s_u^{\gamma} = \bot$. We have then $s_{u}^{\gamma'} = \top$.

Now let's consider $v \in N(u)$. By definition of $I_{\gamma}$, $s_{v}^{\gamma}=\bot$, and $v$ has $u$ as a neighbour that has value $\top$ in $\gamma$. Then the only rule that may be enabled on $v$ in $\gamma$ is \textbf{Refresh}, and we have $s_{v}^{\gamma'}=\bot$.

Thus, $u \in I_{\gamma'}$.
\ \end{proof}

 \setcounter{lemma}{9}
 
\begin{lemma}
From any degree-stabilized configuration $\gamma$, Algorithm  is self-stabilising for  a configuration $\gamma'$ where $I_{\gamma'}$ is a maximal independent set of $V_2 \cup I_{\gamma'}$, with time complexity  $\max\left( -\alpha^2\ln p,\frac{\sqrt{2}}{\sqrt{2}-1}\frac{n}{\alpha} \right)$ rounds
with probability at least $1-p$.
\end{lemma}

\begin{proof}
Consider a degree-stabilized configuration $\gamma_0$, and $\Omega$ the set of all complete executions of the algorithm starting in configuration $\gamma_0$. The probability measure $\mathbb{P}$ is the one induced by the daemon choosing the worst possibility for us at every step.

Consider for $i \in \Nat$ the random variables $X_i$ that denotes the configuration after the $i$-th round has ended ($X_0$ is the constant random variable of value $\gamma_0$).

Consider $(\mathcal{F}_i)_{i\in \Nat}$ the natural filtration associated with $X_i$


Consider the function $f : \gamma \mapsto |I_{\gamma}|$.

$Y_i = \mathds{1}_{f(X_i) - f(X_{i-1}) > 0}$ is the random variable with value $1$ if the size of $I$ increased in the $i$-th round, else $0$.

Consider the stopping time $\tau$ (random variable describing the number of rounds the algorithm take to stabilize on $V_1$) defined by: $$\tau(\omega) = \inf\quickset{n\in \Nat | I \text{ does not change after round } n \text{ in execution } \omega}$$

As $Y_i$ has values in $\quickset{0;1}$, we have $\mathbb{P}(Y_i=1 | \mathcal{F}_{i-1}) = \mathbb{E}[Y_i | \mathcal{F}_{i-1}]$. Also, from Lemma~\ref{bz-lem-advance} we get $\mathbb{P}(Y_i=1 | \mathcal{F}_{i-1}) \geq \alpha \cdot \mathds{1}_{\tau \geq i-1}$. Thus combining the two relations we get:

\begin{equation} \label{eqincr}
\mathbb{E}[Y_i | \mathcal{F}_{i-1}] \geq \alpha \cdot \mathds{1}_{\tau \geq i-1}
\end{equation}

Consider $S_i = \displaystyle \sum_{k=1}^{i}Y_k$ the random variable representing the number of rounds where there have been an increment. Since this cannot happen more time than there are nodes, we get:

\begin{equation} \label{eqS}
S_i \leq n
\end{equation}

Consider $A_i = \displaystyle \sum_{k=1}^{i} \mathbb{E}[Y_{k} | \mathcal{F}_{k-1}]$ the random variable representing the sum of the expected values of the increments at each step.

When $\tau > i$, every for very value of $k \in \segment{1}{i}$ we have $\mathds{1}_{\tau \geq k-1} = 1$. Then using (\ref{eqincr}) we get:

\begin{equation} \label{eqA}
\tau > i \Rightarrow A_i \geq i\alpha
\end{equation}

Consider then the random variable $M_i = \sum_{k=1}^{i} Y_k - \mathbb{E}[Y_k|\mathcal{F}_{k-1}]$ (do note that it is the same as the difference $S_i - A_i$).
\begin{align*}
\mathbb{E}\left[ M_{i+1}|\mathcal{F}_i \right] &= \mathbb{E} \left[ \sum_{k=1}^{i+1} Y_k - \mathbb{E}[Y_k|\mathcal{F}_{k-1}] \middle| \mathcal{F}_i \right] \\
&= \mathbb{E} \left[ M_i + Y_{i+1} - \mathbb{E}[Y_{i+1}|\mathcal{F}_{i}] \middle| \mathcal{F}_i \right] \\
&= \mathbb{E} \left[ M_i \middle| \mathcal{F}_i \right]+ \mathbb{E} \left[Y_{i+1}\middle| \mathcal{F}_i \right] - \mathbb{E} \left[\mathbb{E}[Y_{i+1}|\mathcal{F}_{i}] \middle| \mathcal{F}_i \right] \\
&= M_i + \mathbb{E} \left[Y_{i+1} \middle| \mathcal{F}_i \right] - \mathbb{E} \left[Y_{i+1} \middle| \mathcal{F}_i \right] \\
&= M_i
\end{align*}

Thus $(M_i)_{i\in \Nat}$ is a martingale with respect to the filtration $(\mathcal{F}_i)_{i\in \Nat}$.

We also have $|M_{i+1} - M_i| = |Y_{i+1} - \mathbb{E}[Y_{i+1} | \mathcal{F}_i]| \leq \max(Y_{i+1},\mathbb{E}[Y_{i+1} | \mathcal{F}_i]) \leq 1$.

Thus by Azuma inequality:
\begin{equation}\label{eqM}
\forall \beta \leq 0, \mathbb{P}(M_i \leq \beta) \leq e^{-\frac{2\beta^2}{i}}
\end{equation}

Then using (\ref{eqS}) and (\ref{eqA}) we get $\; \tau >i \Leftrightarrow \tau>i \wedge A_i\geq i\alpha \wedge S_i \leq n \;$. If we drop the first fact and combine the other two of the right hand side we get the implication $\tau >i \Rightarrow S_i - A_i \leq n - i\alpha $ \emph{i.e.}:

$\tau >i \Rightarrow M_i \leq n - i\alpha $

Thus we have $\mathbb{P}(\tau >i) \leq \mathbb{P}(M_i \leq n - i\alpha)$ and for $i \geq \frac{n}{\alpha}$ we can apply (\ref{eqM}) to get :

$\mathbb{P}(\tau >i) \leq e^{-\frac{2(n-i\alpha)^2}{i}}$

For $i\geq \frac{\sqrt{2}}{\sqrt{2}-1}\frac{n}{\alpha}$ (it implies that $i \geq \frac{n}{\alpha}$) we have $\frac{1}{2}(i\alpha)^2 \leq (n-i\alpha)^2$, which give for such $i$ :

$\mathbb{P}(\tau > i) \leq e^{-i\alpha^2}$

For $i \geq -\alpha^2\ln p $, we have $e^{-i\alpha^2} \leq p$

Mixing the two above inequalities, when $i\geq \max\left( -\alpha^2\ln p,\frac{\sqrt{2}}{\sqrt{2}-1}\frac{n}{\alpha} \right)$, we get:

$\mathbb{P}(\tau > i) \leq p$

Which concludes the proof.
 \end{proof}


\end{document}